\documentclass[a4paper,USenglish]{lipics-v2021}
\pdfoutput=1
\hideLIPIcs

\title{FinWhale: An Optimally Resilient Two-Round Terminating DAG Protocol} %TODO Please add

\author{Razya Ladelsky}
{Technion, Israel}
{ladelskyr@cs.technion.ac.il}
{}
{}

\author{Roy Friedman}{Technion, Israel}{roy@cs.technion.ac.il}{}{}

\authorrunning{R. Ladelsky and R. Friedman}

\Copyright{Razya Ladelsky and Roy Friedman}


 %TODO mandatory; please add comma-separated list of keywords

\category{} %optional, e.g. invited paper

\relatedversion{} %optional, e.g. full version hosted on arXiv, HAL, or other respository/website
\nolinenumbers %uncomment to disable line numbering

\EventEditors{Keren Censor-Hillel}
\EventNoEds{1}
\EventLongTitle{40th International Symposium on Distributed Computing (DISC 2026)}
\EventShortTitle{DISC 2026}
\EventAcronym{DISC}
\EventYear{2026}
\EventDate{November 9--13, 2026}
\EventLocation{Rome, Italy}
\EventLogo{}
\SeriesVolume{40}
\ArticleNo{}
\usepackage[ruled]{algorithm}
\usepackage[noend]{algpseudocode}
\newif\ifthesis

\thesisfalse % paper submission version

\begin{document}

\maketitle


\begin{abstract}
  DAG-based Byzantine Fault Tolerant (BFT) protocols provide high-throughput consensus under partial synchrony, but existing DAG protocols still require at least three message delays to commit decisions. 
In contrast, fast-path BFT protocols can achieve optimal two-message-delay termination under favorable conditions, though they do not naturally extend to DAGs.

We present FinWhale, the first DAG-based BFT protocol with a two-message-delay fast path. 
FinWhale extends Mysticeti with a novel fast-path commit mechanism that safely coexists with the protocol’s original slow-path rules. 
To preserve safety across different local DAG views, we introduce new commit structures based on FP-evidence blocks, enabling validators to combine fast-path and slow-path reasoning consistently.

FinWhale operates in the partially synchronous model with $n=3f+2p-1$ validators, matching the known lower bound for fast Byzantine consensus. 
The protocol tolerates up to $f$ Byzantine faults and achieves fast termination whenever at most $1\leq~p\leq~f$ validators fail during the fast path.
Our results show that optimal-latency fast paths can be integrated into uncertified DAG consensus protocols.

\begin{comment}

We introduce \sysnamec, the first DAG-based Byzantine consensus protocol to achieve the lower bounds of latency of 3 message rounds.
Since \sysnamec is built over DAGs it also achieves high resource efficiency and censorship resistance.
% 
\sysnamec achieves this latency improvement by avoiding explicit certification of the DAG blocks and by proposing a novel commit rule such that every block can be committed without delays, resulting in optimal latency in the steady state and under crash failures.
%
We further extend \sysnamec to \sysnamefpc, which incorporates a fast commit path that achieves even lower latency for transferring assets. Unlike prior fast commit path protocols, \sysnamefpc minimizes the number of signatures and messages by weaving the fast path transactions into the DAG. This frees up resources, which subsequently result in better performance.
%
We prove the safety and liveness in a Byzantine context. We evaluate both \sysname protocols and compare them with state-of-the-art consensus and fast path protocols to demonstrate their low latency and resource efficiency, as well as their more graceful degradation under crash failures. \sysnamec is the first Byzantine consensus protocol to achieve WAN latency of 0.5s for consensus commit while simultaneously maintaining state-of-the-art throughput of over 200k TPS. Finally, we report on integrating \sysnamec as the consensus protocol into \realblockchain, resulting in over 4x latency reduction.
\end{comment}

\end{abstract}
\section{Introduction} \label{sec:overview}

Directed Acyclic Graph (DAG)-based protocols aim to provide a Byzantine Fault Tolerance (BFT) replicated ledger, or blockchain, abstraction, i.e., to maintain a consistent total order of transactions across distributed validators, even in the presence of malicious and faulty nodes~\cite{shoal++, mysticeti,  HashGraph, narwhal-tusk, Aleph, mahi-mahi, dag-rider,Cordial-Miners,  starfish, sailfish,  bullshark, shoal}.
DAG-protocols are claimed to offer higher throughput even under asynchrony and failures compared to classical leader-based protocols.
However, many DAG protocols suffer from high latency. 
Mysticeti is a recent DAG-based BFT protocol that can achieve the latency lower bound of three message delays and maintains low latency even in the presence of certain crash failures.
As an indication of its practical relevance, Mysticeti has been recently adopted by both the Sui and IOTA blockchains.

Interestingly, it was shown that BFT consensus can be accelerated with a \emph{fast path}, which allows validators to commit values in just two message delays under favorable conditions. 
This matches the optimal latency achievable in crash fault-tolerant systems. 
Early work by Martin and Alvisi, who introduced the FaB protocol~\cite{fab}, presented a two-round termination algorithm built on PBFT~\cite{PBFT} with \( n \ge 5f + 1 \) validators, along with a parameterized variant requiring \( n = 3f + 2p + 1 \), where \( f \) denotes the total number of faulty validators tolerated, and \( p \le f \) represents the number of faulty validators that need not participate for the fast path to succeed.

Subsequent works by Kuznetsov et al.~\cite{revisiting-optimial-resilience} and Abraham et al.~\cite{good-case-broadcast}
further reduced these thresholds to $n = 5f - 1$ and $n = max(3f+2p-1, 3f+1)$ (for $ p \le f$), 
which have been shown to constitute lower bounds for fast-path BFT partially synchronous protocols.
Banyan~\cite{banyan} introduced a fast path that achieves fast termination in just two message delays, while departing from traditional PBFT-style designs: it is built on the Internet Computer Consensus (ICC) protocol~\cite{ICC} and operates with $n = 3f+2p-1$ validators, where $p \ge 1$ . 

In this work, we address the question of whether decisions within two message delays can be achieved in DAG-based protocols.  
In DAG protocols, the graph continuously grows as rounds advance, and the challenge lies in identifying commit patterns within an already constructed DAG. 
Our approach integrates fast-path commit patterns with those used for slow-path commits, enabling both paths to coexist consistently within the DAG\footnote{The Mysticeti paper~\cite{mysticeti} introduces Mysticeti-FPC, which provides a fast path for transactions that do not require consensus. In contrast, our work adds a fast path to the consensus protocol itself, covering all types of transactions.}.

\noindent \textbf{Our Contribution:} In this work, we introduce \textit{FinWhale}, a 
fast-path mechanism for the Mysticeti protocol, making it the first DAG-based protocol to support termination in just two message delays. 
We prove both safety and liveness under partial synchrony with $n = 3f + 2p - 1$ validators, matching the lower bound. 
The protocol tolerates up to $f$ Byzantine failures and achieves fast termination when the number of actual faults does not exceed $p$ ($1 \le p \le f$), the round leader is honest, and the system has passed GST. This is while preserving the core guarantees and benefits of Mysticeti.

The rest of the paper is organized as follows. 
Section~\ref{sec:model} describes the system model and problem definition. 
Section~\ref{sec:Mysticeti} reviews the Mysticeti protocol. 
Section~\ref{sec:FinWhale} presents the FinWhale protocol. 
An execution example of FinWhale is illustrated in Section~\ref{sec:run_example}. 
Formal safety and liveness proofs are provided in Section~\ref{sec:security} and Section~\ref{sec:more-security}.
Section~\ref{sec:related-work} discusses related work, and Section~\ref{sec:conclusion} concludes the paper. 
\section{Model}
\label{sec:model}
As is common in many BFT works, we assume a message-passing network among a set of $n$ validators, connected via reliable and authenticated point-to-point links. 
Messages sent by honest validators are eventually delivered and cannot be forged, duplicated, or altered by the adversary. 
The adversary is computationally bounded and cannot break standard cryptographic primitives such as digital signatures or hash functions.
In particular, validators cannot impersonate each other.

The network operates under the \emph{partially synchronous} model~\cite{dls88}. 
That is, in each run of the system, there exists a Global Stabilization Time (GST), which is unknown to the validators.
Before GST, message delays are arbitrary and unbounded.
After GST, the network becomes synchronous: all messages between honest validators are guaranteed to be delivered within a known bounded delay~$\Delta$.

We assume that the system can have up to $f$ \emph{Byzantine validators}, i.e., validators that may behave arbitrarily, including maliciously, by sending conflicting or incorrect messages to disrupt the protocol. The remaining validators are considered \emph{honest}, meaning they follow the protocol exactly as specified.

The total number of validators participating in the protocol is
$n \geq 3f + 2p - 1$,
where $p$ is a parameter satisfying $1 \le p \le f$. 
Intuitively, $p$ serves as a flexible threshold for enabling a \emph{fast-path decision}: if the number of faults in the system is at most $p$, the protocol can potentially reach consensus quickly using the fast path. However, even when the number of Byzantine faults exceeds this threshold, the protocol still tolerates up to $f$ Byzantine faults overall by relying on a slower, more conservative decision process to ensure safety and liveness.
Setting $p=1$ enables the protocol to operate with $3f+1$ validators, achieving the optimal resilience bound~\cite{dls88}. In this setting, the protocol takes the fast path after GST under an honest leader provided that at most one validator fails. In contrast, setting $p=f$ guarantees that, under an honest leader after GST, the fast path is always taken.

As in prior DAG works~\cite{mysticeti,dag-rider,bullshark}, the protocol's goal is to solve the \emph{Byzantine Atomic Broadcast} (BAB) problem.
Specifically, a validator $v_k$ can invoke $\mathsf{a\_bcast}_k(m, id)$ to broadcast a message $m$ with a unique sequence number $id \in \mathbb{N}$. 
Each validator $v_i$ outputs delivered messages via $\mathsf{a\_deliver}_i(m, id, v_k)$, indicating that $v_i$ has delivered the message $m$ associated with $id$ that was broadcast by $v_k$.
A protocol solving the BAB problem must satisfy the following properties:

\begin{enumerate}
    \item \textbf{Agreement:} If an honest validator $v_i$ outputs $\mathsf{a\_deliver}_i(m, id, v_k)$, then every other honest validator $v_j$ eventually outputs $\mathsf{a\_deliver}_j(m, id, v_k)$.
    
    \item \textbf{Integrity:} For each sequence number $id \in \mathbb{N}$ and validator $v_k$, there exists at most one message $m$ such that an honest validator $v_i$ outputs $\text{deliver}_i(m, id, v_k)$.
    
    \item \textbf{Validity:} If an honest validator $v_k$ calls $\mathsf{a\_bcast}_k(m, id)$, then every honest validator $v_i$ eventually outputs $\mathsf{a\_deliver}_i(m, id, v_k)$.

    \item \textbf{Total Order:} 
    If an honest validator $v_i$ outputs $\mathsf{a\_deliver}_i(m, id ,v_k)$ 
    before $\mathsf{a\_deliver}_i(m', id', v_k')$, 
    then no honest validator $v_j$ outputs $\mathsf{a\_deliver}_j(m', id', v_k')$ 
    before $\mathsf{a\_deliver}_j(m, id, v_k)$.
\end{enumerate}

\section{Mysticeti Overview} \label{sec:Mysticeti}
The Mysticeti paper~\cite{mysticeti} refers to the core DAG-based atomic broadcast protocol as Mysticeti-C.
For brevity, in the rest of this work, we simply refer to both the system and the protocol as Mysticeti.

Below, we present the basic building blocks of Mysticeti, which also constitute the slow path of our protocol.
The original construction as described in~\cite{mysticeti} was shown to suffer from liveness issues~\cite{liveness-issues, sailfish}, which were subsequently resolved in~\cite{starfish}.
We therefore adopt the version of Mysticeti as presented in~\cite{starfish} , which augments the protocol with a push-based pacemaker and provides provable liveness guarantees. 
%The details of this mechanism are described later.

\noindent \textbf{DAG structure:} The protocol proceeds in consecutive rounds. 
In each round $r$, every honest validator broadcasts a single signed block containing its identifier (author), the round number $r$, payload (users' transactions), and a set of $n$ edges to the latest blocks created by distinct validators in rounds up to $r-1$. 
Among these edges, at least $n - f$ refer to blocks created in round $r-1$.

Each edge includes the hash of a previously received block, and the referenced block is called a \emph{parent} of the new block.
In addition, every block includes an edge that references the previous block created by the same validator. 
%% All the blocks known to a validator form a directed acyclic graph (DAG). %% repeated below

A block is said to be \emph{valid} if it conforms to the above block structure and its signature is correctly verified under the public key of its author. 
Each honest validator maintains a local DAG consisting of the blocks it has created and the valid blocks it has received.
%Only valid blocks are inserted into the local DAG. 
A block is inserted only when all its parent blocks are present; otherwise, it is buffered until the missing parents are received.
The conditions under which a validator creates a block and advances to the next round are defined below when we discuss the pacemaker.

When an $\text{a\_bcast}_i(m, id)$ is invoked at validator $v_i$, the pair $(m,id)$ is pushed to the DAG layer.
Validator $v_i$ will include $(m, id)$ in the payload field of its next created~block.

%All the blocks known to a validator form a directed acyclic graph (DAG), and each honest validator maintains its own local DAG.

\noindent \textbf{Leader slots:}
The protocol deterministically designates a \emph{leader} for each round, according to a round-robin schedule over the validators. Each leader together with its round defines a \emph{leader slot}.  
Note that if the leader equivocates, multiple blocks may be associated with the leader slot.

%Mysticeti defines three possible states for a leader slot: \texttt{to-commit}, \texttt{to-skip}, and \texttt{undecided}. 
Each leader slot is initially marked as \texttt{undecided}. 
Each validator analyzes its local DAG and decides, for every leader slot, whether its state should be changed to \texttt{to-commit} or \texttt{to-skip}.
Once taken, such decisions are non-reversible.
%A slot marked \texttt{undecided} has not yet been assigned a final decision.

A slot marked \texttt{to-commit} designates its associated leader block as \texttt{committed}, and the decision rules ensure that at most one block is committed per slot. 
The committed leader blocks determine the ordering of all blocks in the DAG.

\begin{comment}
    
A slot marked as \emph{to-commit} indicates that the leader block  becomes part of the \emph{commit sequence}, which is the ordered list of leader blocks that determines the ordering of all blocks in the DAG, as explained~below.
A slot marked as to-commit 

A Byzantine round leader may produce multiple versions of its block, in which case at most one of them (or none) can be marked as to-commit. 
A slot marked as to-skip means that none of the leader’s blocks will ever be committed. 
A slot in the undecided state will be resolved later, either to to-commit or to-skip
%—and no subsequent leader can be added to the commit sequence until this decision is made.

\noindent \textbf{Patterns:}  
The DAG is analyzed for specific patterns, the presence of which in turn determines the state of each leader slot.  
Throughout the rest of this paper, we use the convention that whenever a pattern requires a specific number of blocks satisfying a given property, only blocks from distinct validators are considered.
\end{comment}

We now define the notion of voting, a key concept in identifying DAG-based patterns underlying the decision rules.

\noindent \textbf{Voting:}
A block $b'$ in round $r+1$ \emph{votes for} a leader block $b$ in round $r$ if $b$ is one of its parents.
Clearly, a block created by a Byzantine validator may reference more than one leader block; in this case, it votes for the first leader block in its parent set.
\begin{comment}
\noindent
\textbf{Remark.} 
\sysnamec introduces the notion of a \emph{wave}, defined as a sequence of three or more consecutive rounds, after which a decision is made regarding the proposer slots of the first round in the wave. 
The paper discusses the trade-offs associated with different wave lengths in Sections~??. 
In practice, the wave length is fixed to three rounds. 
In this work, we adopt the same configuration, which allows us to employ a simplified definition of support. 
For the same reason, the following definition of a \emph{SP-certificate} (called a certificate in Mysticeti) is also presented in its simplified form.
\end{comment}

Since Mysticeti serves as the slow path of our protocol, we retain its certificate-based decision mechanism, and refer to it as the \emph{slow-path certificate (SP-certificate)}, and refer to its skip condition as the \emph{slow-path skip pattern (SP-skip)}.

\noindent \textbf{SP-certificate:}
Given a leader block $b$ at round $r$, a block $b'$ at round $r+2$ is called an \emph{SP-certificate} for $b$ if it references a quorum of $\lceil \frac{n+f+1}{2} \rceil$ blocks from distinct validators that vote for $b$. 
%\textcolor{red}{-should I mention what a quorum is? No}

\noindent \textbf{SP-skip:}
An \emph{SP-skip pattern} occurs when, for each block $b$ of the leader slot (there may be multiple such blocks if the leader equivocates), there exists a quorum of blocks from distinct validators in round $r+1$ that do not vote for $b$. Particularly, if a block $b$ of the leader slot is not in the local DAG, then all observed blocks in round $r+1$ trivially do not vote for $b$.

Note that in case of conflicting leader blocks for the same slot, at most one of them can have an SP-certificate, since by quorum intersection, at most one can garner the required quorum of blocks from distinct validators that vote for it.
This property is crucial for the correctness of the decision rules.

\noindent
\textbf{Decision rules:}
Each validator attempts to decide each leader slot by first applying the direct decision rule, and then applying the indirect commit rule to earlier undecided slots. 
This process proceeds from the highest  round to the lowest undecided round.

The \emph{direct decision rule} marks a leader slot at round~$r$  as \texttt{to-commit} if there exists a quorum of SP-certificates for its leader block at round~$r+2$, in which case that block is \emph{committed}(see Fig.~\ref{fig:SP-direct}). 
The slot is marked \texttt{to-skip} if the DAG contains an SP-skip pattern for the slot.

The direct commit rule enables committing leader blocks within three message delays when there are no failures and the system is fully synchronized. 
The direct skip rule allows slots to be decided quickly when the leader has crashed or failed to produce a valid block, thereby reducing the number of undecided slots.

If the direct decision rule does not decide the leader slot at round~$r$, the validator applies the \emph{indirect decision rule}:
The validator first searches for an \emph{anchor}. The anchor is the earliest leader slot with round number greater than $r+2$ that is marked as either \texttt{to-commit} or \texttt{undecided}. 
If the anchor is \texttt{undecided}, then the slot remains \texttt{undecided}. 
If the anchor is marked \texttt{to-commit}, the validator checks whether there exists a path from the anchor to an SP-certificate for a block from the leader slot. If such a path exists, the slot is marked \texttt{to-commit} and the corresponding block is \emph{committed} (see Fig.~\ref{fig:SP-indirect-commit}); otherwise, it is marked \texttt{to-skip}.

\noindent \textbf{Commit sequence:}
After applying the decision rules to all leader slots, the validator extends its \emph{commit sequence} in ascending order of round numbers. Leader slots marked as \texttt{to-skip} are omitted, while committed blocks from slots marked as \texttt{to-commit} are included in the sequence.
The sequence terminates at the first \texttt{undecided} slot.
The safety and liveness proofs of Mysticeti establish that all honest validators commit a consistent sequence of leader blocks and that, eventually, every slot is~decided.

\noindent \textbf{Total order:}
Validators obtain a total order over all blocks in the DAG by deterministically sorting the causal histories of leader blocks in the commit sequence.
Specifically, for each leader block in the commit sequence, the latter part of its causal history that has not been already ordered by previous leader blocks is deterministically sorted and added to the total ordering sequence.
\begin{comment}
If a slot contains multiple blocks due to equivocation, exactly one block is (deterministically) selected.
\end{comment}
For each ordered block $b$, validators deliver $(m,id,b.author)$ tuples derived from the $(m,id)$ pairs in the payload of $b$. A tuple is delivered only if no tuple with the same $id$ and $b.author$ has been delivered previously.

\noindent\textbf{Safety intuition:}
If a validator commits a block $b$ for a given slot, then no other validator can skip that slot or commit a conflicting block for the same slot.
To that end, a key property proved in \cite{mysticeti,starfish} and by our Lemma~\ref{lemma:recursion-safety_a} states that: if $b$ is committed directly by some validator, then a path from the anchor to an SP-certificate for $b$ exists in every validator’s DAG. 
This property rules out indirect skipping by definition.
Moreover, the presence of an SP-certificate, whether $b$ is committed directly or indirectly, excludes both direct skipping and committing a conflicting block, by quorum intersection.
The remaining case is when $b$ is committed indirectly by different validators.
In this case, all validators agree on the same committed anchor block, as shown in~\cite{mysticeti,starfish} and by our Lemma~\ref{lemma:consistent-slots-1}.
Consequently, they observe the same causal history and reach the same decision for that slot.

\noindent\textbf{Push Pacemaker:}
The Push Pacemaker is responsible for ensuring liveness, by specifying the conditions for advancing rounds, creating blocks, and broadcasting unknown parts of the history.
It consists of the following aspects:
%We adopt the conditions with their notations as follows:

\noindent\textbf{Advancing rounds:}
A validator advances from round $r-1$ to round $r$ when the following two conditions are satisfied:
\begin{description}
    \item[\textbf{A1:}] There are $n-f$ round $r-1$ blocks from distinct validators in the local DAG.
    \item[\textbf{A2:}] The validator has created its own block in round $r-1$.
\end{description}

\noindent\textbf{Creating blocks:}
After advancing to round $r$, a validator records its local time.
We set $\delta_{LT} = 2\Delta$ as the timeout duration.
The validator then creates its block when one of the following three conditions holds:
\begin{description}
  \item[\textbf{C1:}] The following two requirements are satisfied:
    \begin{description}
      \item[\textbf{L1:}] The local DAG contains a leader block from round $r-1$.
      \item[\textbf{L2:}] The local DAG contains either (i) a quorum of voters from
        distinct validators for the leader of $r-2$, or (ii) an SP-skip
        pattern for the leader of $r-2$.
    \end{description}

  \item[\textbf{C2:}] The $\delta_{LT}$ timeout has expired.

  \item[\textbf{C3:}] The local DAG contains $n-f$ blocks from distinct validators
    in round $r$.
\end{description}

When \textbf{C1} triggers block creation, the validator should choose to include in its parents the set of blocks that were used to satisfy requirement \textbf{L2} (it is important in cases where there are multiple versions of $r-1$ blocks belonging to Byzantine validators).
Waiting for \textbf{C1} to be satisfied ensures that after GST, honest leaders get committed.
\textbf{C2} requires that upon expiration of a timeout, the validator creates a block, even if condition \textbf{C1} is not satisfied.
\textbf{C3}  ensures that slower validators can catch up: if \textbf{C1} is not satisfied but $n-f$ blocks have already been produced in round $r$, a validator may proceed without waiting for the timeout to expire.

\noindent\textbf{Broadcasting unknown history:}
When a validator receives a block, it attempts to add it to its local DAG only if all its parents (and, transitively, its entire causal history) are already present in the DAG. Since Byzantine validators may send blocks to only a subset of validators or send conflicting versions to different validators, when an honest validator creates and sends its block, its parents may include blocks created by Byzantine validators. 
To ensure that such blocks can be added to other validators’ DAGs, validators follow the principle of cordial dissemination~\cite{Grassroots}, by which a validator $v_i$ sends to another validator $v_j$ the blocks in its local DAG that it believes $v_j$ did not yet receive.
Formally, let $\mathrm{DAG}_i$ denote the set of blocks in the local DAG of $v_i$. From the perspective of $v_i$, the set of blocks known to $v_j$ is defined as $\mathrm{known}_i(j) = H_{i,j} \cup S_{i,j}$, where $H_{i,j}$ is the union of the causal histories of all blocks in $\mathrm{DAG}_i$ created by $v_j$, and $S_{i,j}$ is the set of all blocks sent from $v_i$ to $v_j$. Accordingly, the set of blocks unknown to $v_j$ is $\mathrm{unknown}_i(j) = \mathrm{DAG}_i \setminus \mathrm{known}_i(j)$.

Upon any of the following pacemaker triggers, a validator $v_i$ sends the unknown portion of its history to its peers:
\begin{description}
    \item[\textbf{B1:}] creating a block.
    \item [\textbf{B2:}] advancing to a new round.
\end{description}

The Push pacemaker ensures synchronicity after GST: when a fast honest validator advances to a new round, its history is broadcast to all validators via \textbf{B2}. As a result, they can create blocks in all rounds up to that round via \textbf{C3} and advance to the new round according to \textbf{A1} and \textbf{A2}. 
Moreover, \textbf{B1} ensures that if a fast honest validator creates a block under \textbf{C1}, then all honest validators observe the same conditions and can therefore also create their blocks under \textbf{C1}.

\noindent\textbf{Liveness intuition:}
With the Push pacemaker, \cite{starfish} establishes synchronicity after GST: validators enter each new round within $\Delta$ and create blocks within $\Delta$ of one another. This property was assumed without proof in \cite{mysticeti}. Building on this, \cite{mysticeti,starfish} show that, after GST, any honest leader slot is eventually marked \texttt{to-commit} by the direct decision rule, and that all leader slots are eventually decided. Consequently, every honest leader block is included in the commit sequence. Moreover, any block created by an honest validator eventually becomes part of the causal history of some committed honest leader block, and is therefore ordered and delivered.

\section{\textit{FinWhale}}
\label{sec:FinWhale}
\subsection{Protocol Description}
In addition to the Byzantine Atomic Broadcast properties, FinWhale satisfies an additional \emph{Fast Termination} property.
We adapt its definition from Banyan~\cite{banyan} to our setting.
\begin{comment}
Martin and Alvisi [23] present FaB Paxos– a fast Byzantine consensus protocol with $n=5f+1$. Moreover, they present a parameterized version of the protocol: it runs on $n=3f+2t+1$ processes ($t\leqf$), tolerates $f$ Byzantine failures, and is able to commit after just two steps in the common case when the leader is correct, the network is synchronous, and at most $t$ processes are Byzantine.
\end{comment}

\begin{definition}
If the network is momentarily synchronous, the leader is honest, and at least $n - p$ validators behave momentarily honestly, then the leader block is committed (its slot is marked \texttt{to-commit}) within two message delays.
\end{definition}

To support the Fast Termination property, we introduce an additional direct commit rule: A leader block $b$ of round $r$ is \emph{fast-directly committed} if there exist $n-p$ blocks from distinct validators in round $r+1$ that reference $b$, i.e., if there are $n-p$ voters for $b$. 
With this rule, a leader block can be committed within two message delays when network conditions permit.

In Mysticeti, the presence of a direct commit pattern within a given DAG guarantees that any other DAG will contain a corresponding pattern serving as evidence that neither a skip can be determined for the slot nor a conflicting block can be committed.
Similarly, we require that if a fast direct commit of the leader block occurs in the DAG of some validator, then every other validator must observe in its DAG a corresponding pattern that rules out both conflicting commits and skip decisions for the same slot. 

We introduce a new notion, where blocks in round $r+2$ serve as an \emph{FP-evidence} for a leader block $b$ of round $r$.
Intuitively, when a block is committed directly via the fast path in a DAG of validator $v_i$ during round $r+1$, it leaves evidence in round $r+2$ in all other DAGs in the form of FP-evidence blocks.

We note that a single FP-evidence block represents a weaker form of voting evidence than an SP-certificate, and the DAG may simultaneously contain FP-evidence for conflicting leader blocks.
However, the presence of a quorum of such FP-evidence blocks suffices to guarantee consistent behavior across all honest validators’ DAGs.

%We now turn to the exact details of our protocol.

\noindent \textbf{DAG construction:}
In Mysticeti, advancing from the current round to the next round according to condition \textbf{A1} requires observing $n - f$ blocks from different validators.
To enable the fast path under the tighter bound $n = 3f + 2p - 1$, an honest validator must take into account not only the blocks from the current round, but also the behavior of the leader of the previous round\footnote{
This additional requirement does not arise under the looser bound $n = 3f + 2p + 1$, as discussed later.
}.

Specifically, a Byzantine leader in the previous round may equivocate by broadcasting multiple conflicting versions of its block. 
%Honest validators can detect such misbehavior by inspecting the blocks they receive in the current round.
However, if an honest validator receives blocks that reference more than one version of the leader block, then this is an evidence that the previous round leader is Byzantine. 
In this case, the validator can safely wait for $n-f$ blocks not including the block authored by the  Byzantine leader without blocking.
%(i.e., if this block is received, and there are $n-f$ blocks from distinct validators, the validator needs to wait for exactly one more block).
Consequently, the validator observes and references at least $n - f$ blocks, among which at most $f - 1$ are from Byzantine validators, because the block authored by the detected equivocating leader is not included as a parent when the validator constructs its next~block.
%A block~$b$ at round~$r+1$ is said to \emph{vote} for a leader’s block~$b'$ from round~$r$ if it includes an edge to~$b'$ (or rather a reference to its hash).

Throughout this section, we denote $L_r$ the leader of round $r$.
In particular, $L_{r-1}$ and $L_{r-2}$ denote the leaders of rounds $r-1$ and $r-2$, respectively.

\noindent\textbf{Leader-consistent:}
A set of round-$r$ blocks $B_r$ is said to be \emph{leader-consistent} with respect to $L_{r-1}$ if all blocks in $B_r$ vote for at most one block proposed by $L_{r-1}$.

%Next, we define the conditions under which a block is considered valid. 
%As before, only valid blocks are inserted into the DAG.
%In particular, even blocks produced by Byzantine validators must satisfy the block validity rules.
%%%% The above is redundant.

\noindent \textbf{Block validity:}
%We require that all blocks obey the block validity rules; otherwise, they are not added to the DAG of an honest validator.
As before, blocks are added to the DAG of an honest validator only if it obeys the block validity rules.
FinWhale's definition of a valid block extends Mysticeti's validity rules by adding the following requirements:
\begin{enumerate}
    \item Each block has at most a single parent per validator. We do not allow malformed blocks (e.g., those created by Byzantine validators) that contain multiple edges to blocks of the same validator.
    \item A block in round $r$ must be valid with respect to round $r-2$. That is, it links to at least $n-f$ blocks from round $r-1$, each from a distinct validator, and satisfies one of the following conditions:
    \begin{alphaenumerate}
        \item the set of parent blocks is leader-consistent with respect to $L_{r-2}$, or
        \item the round-$r-1$ block created by $L_{r-2}$ is not included among the parents.
    \end{alphaenumerate}
\end{enumerate}

\noindent
\textbf{FP-evidence, Non-FP-evidence:}
We say that a block $b'$ of round $r+2$ \emph{exposes equivocation by $L_r$} if its parent set is not leader-consistent, i.e., its causal history contains multiple conflicting versions of $L_r$'s block.
%This notion is used in the FP-evidence~definition.
%

A block $b'$ of round $r+2$ is an \emph{FP-evidence} for a leader block $b$ of round $r$ if one of the following conditions holds:
\begin{itemize}
    \item \textbf{Equivocating case:} $b'$ exposes equivocation of $L_r$. It references at least $f+p$ parent blocks that vote for $b$, while fewer than $f+p$ of its parent blocks vote for any block conflicting with $b$.

\item \textbf{Non-equivocating case:} $b'$ does not expose equivocation of $L_r$. It references at least $f+p-1$ parent blocks that vote for $b$.
\end{itemize}

Finally, a block at round $r+2$ is a \emph{Non-FP-evidence} if it is not an FP-evidence for any block proposed by $L_r$. 

Examples of FP-evidence and Non-FP-evidence blocks are in Figure~\ref{fig:FP-evidence}.

\begin{figure*}[t]
    \centering
    \vspace{-0.4cm}

    % Set height once for all figures
    \newcommand{\figH}{5.0cm}

    \begin{subfigure}[t]{0.45\textwidth}
        \centering
        \includegraphics[height=\figH, keepaspectratio]{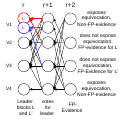}
        \hspace{0.25\textwidth}\textbf{(a)}
       %\caption{}
       %\label{fig:FP-evidence}
    \end{subfigure}
    \hfill
    \begin{subfigure}[t]{0.45\textwidth}
        \centering
        \includegraphics[height=\figH, keepaspectratio]{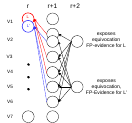}
        \hspace{0.25\textwidth}\textbf{(b)}
        %\caption{}
        %\label{fig:FP-evidence-b}
    \end{subfigure}
    \vspace{-0.2cm}
    \caption{\footnotesize FP-evidence and Non-FP-evidence examples under equivocation. $V_1$ is an equivocating leader in round $r$, and votes in round $r+1$ are colored according to the leader block they support.    
\newline\textbf{(a)}\quad FP-evidence and Non-FP-evidence examples illustrating the cases where equivocation is exposed and not exposed with $f=p=1$. If a block in round $r+2$ exposes the equivocation, it qualifies as FP-evidence for $L$ if it has at least two ($f+p$) parents voting for $L$ and at most one ($f+p-1$) parent voting for $L'$. If the block does not expose the equivocation, a single (f+p-1) parent voting for $L$ is sufficient to qualify as FP-evidence for $L$. Symmetrically, the same conditions apply for $L'$. Any block that neither qualifies as an FP-evidence for L nor for L' is a Non-FP-evidence block. 
\newline\textbf{(b)}\quad FP-evidence for conflicting leader blocks while exposing equivocation, with $f=2, p=1$. Both round-$(r+2)$ blocks expose the equivocation. The round-$(r+2)$ block created by $V_2$ has three ($f+p$) parents voting for $L$ and two ($f+p-1$) parents voting for $L'$, and therefore qualifies as an FP-evidence for $L$. Symmetrically, the round-$(r+2)$ block created by $V_5$ qualifies as an FP-evidence for $L'$.
}
\label{fig:FP-evidence}
\end{figure*}

A key property we prove in Lemma~\ref{lemma:fast-commit-cert} below is that if a leader block of round $r$ is directly committed via the fast path, then every block at round $r+2$ in any DAG is an FP-evidence block for $b$.
 
%We continue to use the previously defined notion of SP-certificate and SP-skip pattern: 
%A block in round $r+2$ is an SP-certificate for a leader block $b$ in round $r$ if it references a quorum of blocks (from distinct validators) voting for $b$.
%An SP-skip pattern occurs when, for each leader block $b$ of round $r$, there exists a quorum of round $r+1$ blocks from distinct validators that do not vote for $b$.
FinWhale relies on the same definitions of SP-certificate and SP-skip patterns as Mysticeti.
Recall that in our setting, a quorum of size $\left\lceil \frac{n+f+1}{2} \right\rceil$ equals $2f+p$.
With the additional definitions of FP-evidence blocks and Non-FP-evidence blocks, we are now ready to formally define FinWhale's decision rules. 
Here again, each validator iterates over rounds from the highest to the lowest undecided round and applies the decision rules.

\noindent
\textbf{Direct decision rule:}  
The validator marks a leader slot of round $r$ as \texttt{to-commit} if it observes \emph{one of} the following cases:
\begin{enumerate}
    \item There exists a quorum of $\left\lceil \frac{n+f+1}{2} \right\rceil = 2f+p$ SP-certificates from distinct validators in round $r+2$ for a block $b$ of the~slot. $b$ is then \emph{committed}.
    \item There exist $n-p$ blocks from distinct validators in round $r+1$ that vote for a block $b$ of the slot. $b$ is then \emph{committed}.
\end{enumerate}

The validator marks a slot as \texttt{to-skip} when \emph{both of} the following conditions hold:
\begin{enumerate}
    \item The SP-skip pattern is observed in round $r+1$.
    \item There are at least $\left\lceil \frac{n+f+1}{2} \right\rceil = 2f+p$ Non-FP-evidence blocks from distinct validators in round $r+2$.
\end{enumerate}

\noindent \textbf{Indirect decision rule:}
If during the reverse pass over undecided round, the direct decision rules did not decide a given slot, the validator tries to apply the indirect decision rules:
First it searches for an anchor, which is defined as in Mysticeti, to be the first slot with round number greater than $r+2$ that is marked either \texttt{to-commit} or \texttt{undecided}.
If the anchor is \texttt{undecided} then the slot is marked \texttt{undecided} as well.
If the anchor is marked \texttt{to-commit}, the validator marks the slot as \texttt{to-commit} if one of the following conditions holds:
\begin{enumerate}
    \item There exists a path from the anchor to an SP-certificate for a block $b$ of the slot, in which case $b$ is \emph{committed}.
    \item There exist paths from the anchor to a quorum of
    $\left\lceil \frac{n+f+1}{2} \right\rceil = 2f+p$ FP-evidence blocks from distinct validators, all for the same block $b$ of the slot, in which case $b$ is \emph{committed}.
\end{enumerate}
Otherwise, the slot is marked \texttt{to-skip}.

The decision rules are illustrated in Figure~\ref{fig:decision-rules}.
Note that both of the above conditions may hold simultaneously for conflicting blocks of the same slot. 
In such a case, one of the blocks is selected for commitment according to a deterministic rule. 
Such a pattern cannot arise if any validator can decide the slot directly, either by direct skip or direct commit: a direct skip rules out the conditions for an indirect commit, and a direct commit rules out the conditions for an indirect commit of a conflicting block. Hence, all validators must decide the slot indirectly, in which case they all reach the same decision. The patterns induced by direct decisions, as well as the consistency of indirect decisions, are described in the safety intuition and proved in the safety section.

\begin{figure*}[t]
    \centering
    \vspace{-0.4cm}

    % Set height once for all figures
    \newcommand{\figH}{4.0cm}

    \begin{subfigure}[t]{0.3\textwidth}
        \centering
        \includegraphics[height=\figH, keepaspectratio]{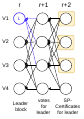}
        \caption{\footnotesize Direct commit via the slow path. There is a quorum of three SP-certificates in round $r+2$, each containing a quorum of votes for the leader block.}
        \label{fig:SP-direct}
    \end{subfigure}
    \hfill
    \begin{subfigure}[t]{0.3\textwidth}
        \centering
        \includegraphics[height=\figH, keepaspectratio]{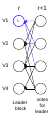}
        \caption{\footnotesize Direct commit via the fast path. There are  three ($n-p =3 $) votes in round $r+1$.}
        \label{fig:FP-direct}
    \end{subfigure}
    \hfill
    \begin{subfigure}[t]{0.3\textwidth}
        \centering
        \includegraphics[height=\figH, keepaspectratio]{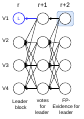}
        \caption{\footnotesize  Direct skip: There is a  quorum of non-voters in $r+1$ and a quorum of Non-FP-evidence blocks for the leader in $r+2$.}
        \label{fig:direct-skip}
    \end{subfigure}

    \vspace{0.4cm}

    \begin{subfigure}[t]{0.3\textwidth}
        \centering
        \includegraphics[height=\figH, keepaspectratio]{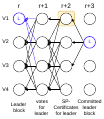}
        \caption{\footnotesize Indirect commit: There is a path from anchor to an SP-certificate for the leader block.}
        \label{fig:SP-indirect-commit}
    \end{subfigure}
    \hfill
    \begin{subfigure}[t]{0.3\textwidth}
        \centering
        \includegraphics[height=\figH, keepaspectratio]{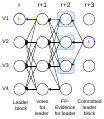}
        \caption{\footnotesize  Indirect commit: There are paths from the anchor to a quorum of FP-evidence blocks for the leader block.}
        \label{fig:FP-indirect-commit}
    \end{subfigure}
    \hfill
    \begin{subfigure}[t]{0.3\textwidth}
        \centering
        \includegraphics[height=\figH, keepaspectratio]{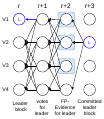}
        \caption{\footnotesize Indirect skip: No path from the anchor to an SP-certificate. No paths to a quorum of FP-evidence blocks for the leader block.}
    \end{subfigure}

    \vspace{-0.2cm}
    \caption{\footnotesize FinWhale decision rules for $f=p=1$. Votes for the leader block are colored in blue. In this setting, there are four validators and the quorum size is $3$.}
\label{fig:decision-rules}
\end{figure*}

\noindent\textbf{Commit sequence:}
The algorithm proceeds similarly to Mysticeti. It examines leader slots from the highest round down to the lowest undecided round. For each slot, the validator applies the commit rules to determine its state.
The commit sequence is then extended in ascending order of round numbers by including slots marked \texttt{to-commit} together with their committed blocks, while omitting slots marked \texttt{to-skip}. The process stops upon reaching the first \texttt{undecided} slot.
Figure~\ref{fig:commit-sequence} illustrates how the commit sequence is extended by including committed leader blocks until reaching the first \texttt{undecided} slot.

\begin{figure}[t]
    \centering    \includegraphics[width=0.6\textwidth]{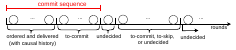}
    \caption{\footnotesize Leader slots over a timeline of rounds (r, r+1, r+2, …); the state is annotated beneath each corresponding interval. The commit sequence is extended by committed leader blocks up to the first undecided slot.}
    \label{fig:commit-sequence}
\end{figure}

\noindent\textbf{Total order:}
As in Mysticeti, the final total order is obtained by deterministically sorting the causal histories of the committed leader blocks. For each ordered block $b$, each $(m,id)$ pair in its payload is delivered as $(m,id,b.\text{author})$, provided that no tuple with the same $id$ and $b.\text{author}$ has been delivered previously.

\noindent\textbf{Safety intuition:}
validators cannot reach conflicting decisions: they either both commit the same block for a leader slot $s$ or both skip $s$.
We rely on the following structural properties of the protocol:
\begin{enumerate}
    \item A slow-path direct commit of block $b$ yields a path from the anchor to an SP-certificate for $b$ in every DAG (Lemma~\ref{lemma:recursion-safety_a}).
    \item A fast-path direct commit yields a path from the anchor to a quorum of FP-evidence blocks for $b$ in every DAG (Lemma~\ref{lemma:recursion-safety_b}).
    \item An SP-certificate for $b$ is also an FP-evidence for $b$ (Lemma~\ref{lemma:SP-FP-relation}).
\end{enumerate}

We claim that if a validator directly skips a leader slot $s$, no other validator can commit a block for the slot (Lemma~\ref{lemma:direct-skip-a}). 
A direct skip requires collecting a quorum of blocks not voting for any block of $s$ together with a quorum of Non-FP-evidence blocks.
By quorum intersection, this precludes the existence of an SP-certificate or a quorum of FP-evidence blocks for any block of slot $s$ in any other DAG. 
Since either of these is necessary for an indirect commit, the latter is impossible.
Further, the presence of a quorum of non-voters for any leader blocks also rules out a direct commit in any other~DAG. 

If a validator directly commits a block $b$ of slot $s$, then no other validator can skip $s$ (Lemma~\ref{lemma:direct-commit-safety-a}) or commit a conflicting block (Lemma~\ref{lemma:unique-cert_2} and Lemma~\ref{lemma:unique-cert_3}).
When a validator directly commits $b$, depending on whether the slow or fast path is taken, any other DAG contains either a path from the anchor to an SP-certificate for $b$ or paths from the anchor to a quorum of FP-evidence blocks for $b$ (Properties~1 and~2, respectively). This already rules out indirect skipping in any other DAG by definition.
To show consistency with respect to a conflicting block $b'$, observe that a direct commit of $b$ implies a quorum of voters for $b$, which by quorum intersection rules out the existence of an SP-certificate for $b'$. 
It also rules out a quorum of FP-evidence blocks for $b'$.
Indeed, if $b$ is committed via the fast path, then by Property~2 every DAG contains a quorum of FP-evidence blocks for $b$. 
If $b$ is committed via the slow path, the SP-certificates for $b$ form a quorum, which in turn induces a quorum of FP-evidence blocks for $b$ by Property~3. 
In both cases, quorum intersection prevents the existence of an FP-evidence quorum for $b'$.
The absence of either an SP-certificate or a quorum of FP-evidence blocks for $b'$ precludes an indirect commit of $b'$.  
Moreover, by a simple quorum intersection argument, a direct commit of $b'$ is also impossible. 

Finally, if validators decide $s$ indirectly, they rely on the same committed anchor block, and therefore observe the same causal history and reach the same decision (Lemma~\ref{lemma:consistent-slots-1}).

\noindent \textbf{Pacemaker conditions:}
Similarly to Mysticeti, we define pacemaker conditions for FinWhale that ensure both liveness and the creation of valid blocks. 
%These conditions closely follow those of Mysticeti, except for condition $\mathbf{A1'}$, which determines advancement from round $r-1$ to round $r$ based on whether $L_{r-2}$ equivocates.

\noindent \textbf{Advancing rounds:}
An honest validator advances from round $r-1$ to round $r$ when two conditions are satisfied: $\textbf{A1'}$ and $\textbf{A2}$. 
Condition $\textbf{A2}$ is inherited from Mysticeti and is restated here for completeness. Condition $\textbf{A1'}$ determines the threshold on the number of round $r-1$ blocks (from distinct validators) that are required for advancement, as in Mysticeti’s $\textbf{A1}$.
However, here the value depends on whether a leader-consistent set is observed.
\begin{description}
    \item[\textbf{A1'}] The local DAG contains either (i) a leader-consistent set of round-$(r-1)$ blocks from $n-f$ distinct validators, or (ii) a set of round-$(r-1)$ blocks from $n-f$ distinct validators that excludes blocks authored by $L_{r-2}$.
    
    \item[\textbf{A2}] The validator has created its own block in round $r-1$.
\end{description}

\noindent \textbf{Creating Blocks:}
%We adopt the same block creation conditions as Mysticeti (C1, C2, C3), which are restated here for completeness. 
A block in round $r$ is created if at least one of the following conditions is satisfied:
%\begin{enumerate}[label=\textbf{C\arabic*}, leftmargin=*]
\begin{description}
    %\item Wait for $L_{r-1}$ and voters for $L_{r-2}$ block:
    \item[\textbf{C1}] The following two requirements are satisfied:
    %\begin{enumerate}[label=\textbf{L\arabic*}, leftmargin=2em]
    \begin{description}
        \item[\textbf{L1}] The local DAG contains a block of $L_{r-1}$. 
        \item[\textbf{L2}] The local DAG contains either $(i)$ a quorum of voters from distinct validators for $L_{r-2}$ block, or $(ii)$  contains an SP-skip pattern for $L_{r-2}$
    \end{description}
    %\end{enumerate}

    \item[\textbf{C2}] The timeout $\delta_{LT} = 2\Delta$ has expired.

    \item[\textbf{C3}] The local DAG contains $n-f$ blocks from distinct validators in round $r$.
\end{description}
%\end{enumerate}

\noindent\textbf{Selecting the parents:}
When a block is created in round $r$, the validator must select its parent set from round $r-1$ blocks. Byzantine validators may produce multiple blocks in round $r-1$, but the validator includes at most one parent per validator, while selecting the blocks that satisfy conditions $\textbf{L1}$ and $\textbf{L2}$  if condition $\textbf{C1}$ triggers block creation.

Whenever the parent set is not leader-consistent with respect to $L_{r-2}$, the round-$(r-1)$ block created by $L_{r-2}$ must be excluded. Accordingly, the parent selection algorithm needs to exclude this block whenever the resulting parent set violates leader-consistency with respect to $L_{r-2}$.
However, when there are  round-$(r-1)$ blocks from exactly $n-f$ validators, excluding this block would leave only $n-f-1$ parents, violating the requirement that a parent set contains at least $n-f$ blocks. Therefore, in this case, we must guarantee that multiple round-$(r-1)$ blocks produced by Byzantine validators do not prevent the construction of a leader-consistent parent set. Condition \textbf{A1'} ensures the existence of such a leader-consistent set of size $n-f$.

Formally, the parent selection proceeds as follows:

\begin{enumerate}
    \item Initialize the candidate set with the blocks in the local view of $DAG[r-1]$.

    \item If the round-$(r-1)$ block proposed by $L_{r-2}$ is present and round-$(r-1)$ blocks from exactly $n-f$ distinct validators are available, construct a leader-consistent set with respect to $L_{r-2}$ that contains at least one block from each of these validators. Such a set exists by condition \textbf{A1'}. Restrict the candidate set to this leader-consistent set.

    \item Construct the parent set from the candidate set such that:
(1) for each validator $v$, the parent set contains exactly one block authored by $v$, and
(2) if condition $\textbf{C1}$ triggers block creation, all blocks in the candidate set that satisfy condition $\textbf{L1}$ and $\textbf{L2}$ are included in the parent set.
    
    \item 
    If the current parent set is not leader consistent w.r.t. $L_{r-2}$, exclude the round-$(r-1)$ parent block proposed by $L_{r-2}$ from the parent set, even if it was previously included (by the previous step).
\end{enumerate}

Finally, the validator includes blocks from rounds older than $r-1$ in its parent set as follows: the latest block from each author that has not yet been referenced by a previous block of the validator is included in the parent set.
This ensures that every block created by an honest validator appears in the causal history of some honest leader block, and is therefore ordered and its payload delivered.
%\todo{weak links}

\noindent\textbf{Broadcasting history:}
A validator broadcasts its unknown history under the same $\textbf{B1}$ and $\textbf{B2}$ conditions as in~Mysticeti.
%\begin{description}
%    \item[$\mathbf{B1}$] Upon creating a block.
%    \item[$\mathbf{B2}$] Upon advancing to a new round.
%\end{description}

\noindent \textbf{Liveness intuition:}
Liveness in FinWhale is based on the liveness of the slow path. After GST, every leader block created by an honest validator is eventually included in the commit sequence, and its causal history is delivered. When conditions allow the fast path to be taken, leader blocks can be marked \texttt{to-commit} more quickly, thereby accelerating total ordering and delivery.

\noindent \textbf{Additional Discussion Points:}
%We described the modifications to the block construction and rules needed to incorporate a fast path under the proven lower bound of $n = 3f + 2p - 1$ validators.  
%Simpler definitions and a simpler construction would have been possible had we relaxed this requirement and instead assumed $n = 3f + 2p + 1$ validators.  
It is possible to simplify FinWhale by relaxing the resilience threshold from the lower bound of $n = 3f + 2p - 1$ validators to $n = 3f + 2p + 1$.
For instance, this would enable a simpler DAG construction, in which it is not necessary to wait for an additional block in round $r+1$ when the round-$r$ leader is proven to be Byzantine.
Moreover, eliminating the need to distinguish between equivocating and non-equivocating cases would simplify the definition of FP-evidence.

Under the relaxed assumption $n = 3f + 2p + 1$, there exist scenarios in which Byzantine leader blocks may be directly committed despite equivocation.
In contrast, under the tighter bound of $n = 3f + 2p - 1$, enforcing that at most $f-1$ Byzantine blocks are selected as parents prevents such direct commits.
These scenarios are rare and arise only under highly specific conditions.

Mysticeti~\cite{mysticeti} targets direct skipping as a mechanism to efficiently handle benign crashed leaders by promptly skipping them and reducing their impact on performance.
When the leader is Byzantine, however, it may attempt to prevent direct skipping, thereby leaving the slot undecided. In particular, to preclude the formation of a quorum of non-voters after GST, a Byzantine leader must ensure that its block is disseminated to sufficiently many honest validators so that they produce voting blocks; otherwise, a quorum of non-voters may form, leading to a direct skip.
On the other hand, once honest validators create voting blocks, they broadcast their history, which may cause other honest validators to receive the leader’s block and form an SP-certificate, resulting in a direct commit. Therefore, a Byzantine leader must carefully and selectively disseminate its block so that enough honest validators vote for it, while ensuring that their history broadcasts do not propagate in time to enable SP-certificate formation. This leaves the leader slot undecided.
In \textsc{FinWhale}, however, direct skipping  requires a quorum of non-FP-evidence blocks in addition to a quorum of non-voters. This makes the Byzantine attack simpler in the optimal-resilience setting. In particular, the leader may withhold its block from all honest validators, while all $f$ Byzantine validators send their round-$r+1$ blocks voting for the leader. When $p=1$, these $f$ Byzantine votes are sufficient to cause honest validators to produce FP-evidence blocks, thereby preventing a direct skip.
As in Mysticeti, when the system is configured with $n = 3f + 2p + 1$ validators, this simple attack is no longer feasible. It remains an open question whether achieving the lower bound on $n$ inherently facilitates this simple attack that prevents direct skipping.

FinWhale differs from Mysticeti in the timing at which the direct-skip decision is made. In our protocol, a direct skip can occur only at the end of round~$r+2$, whereas in Mysticeti, a direct skip may occur already in round~$r+1$.
However, this difference does not actually delay the commit sequence.
The purpose of the direct skip rule is to decide slots as early as possible, so as not to delay the commit of subsequent leaders.
Since in Mysticeti the next leader can be committed only at round~$r+3$, our modified direct skip decision does not introduce any additional delay in the overall commit sequence.

\subsection{Algorithm and Pseudo Code} \label{sec:algorithms}

\begin{algorithm}[t]
\caption{DAG Construction for validator $v_{i}$, Part 1}
\label{alg:construction1}
\footnotesize
\begin{algorithmic}[1]

\State \textbf{Global variables:}  \State $buffer \gets \emptyset$
\State $blocksToPropose \gets \emptyset$ \Comment{Valid blocks of transactions from clients}
\State $\delta_{LT} \gets 2\Delta$
\State $n \gets 3f + 2p - 1$ 
\State $r_{decided} \gets 0$ \Comment{The most recent round in the sequence of decided leaders}
\State $r_{highest} \gets 0$ \Comment{The highest block round observed in the local DAG}
%\State $r_{last} \gets 0$
\State $DAG \gets \Call{GenesisBlocks}{\,}$
\State $\forall j:\; Known[j] \gets \Call{GenesisBlocks}{\,}$

%\Statex
%\Upon {$a\_bcast_i(m, id)$} 
%    \State blocksToPropose.enqueue($m, id$)
%\EndUpon

\Statex
\State
\textbf{upon} $a\_bcast_i(m, id)$
\State blocksToPropose.enqueue$(m, id)$

\Statex
\Procedure{Execute}{}
    \For{$r = 1,2,\ldots$}
        \State \Call{ExecuteRound}{$r$}
    \EndFor
\EndProcedure

\Statex
\Procedure{ExecuteRound}{$r$}
    \State \Call{BroadcastUnknownHistory}{\,}
    \State $\tau_{enter}(r) \gets$ \Call{LocalTime}{\,}
    
    \While{\Call{TryCreateBlock}{$r,\tau_{enter}(r)$} = \textbf{false}}
        \State \Call{ReceiveMessages}{\,}
        \State \Call{TryDecide}{$r_{decided}, r_{highest}$}
    \EndWhile

    %\State $r_{last} \gets r$
    \State \Call{BroadcastUnknownHistory}{\,}

    \While{\Call{TryAdvanceRound}{$r$} = \textbf{false}}
        \State \Call{ReceiveMessages}{}
        \State \Call{TryDecide}{$r_{decided}, r_{highest}$}
    \EndWhile

\EndProcedure

\Statex
\Procedure{BroadcastUnknownHistory}{\,}
    
    \State $DAG_{snap} \gets DAG$

    \For{$j \in \{1,\ldots,n\} \setminus \{i\}$}
        \State $Unknown \gets DAG_{snap} \setminus Known[j]$
        \State $msgs \gets Unknown$
        \State \Call{Send}{$v_{j}, msgs$}
        \State $Known[j] \gets Known[j] \cup \{msgs\}$
    \EndFor

\EndProcedure

\Statex
\Procedure{ReceiveMessages}{\,}
    \State $b \gets \Call{ReceiveNetworkMessage}{\,}$
            \State \Call{TryAddToDAG}{b}

\EndProcedure

\Statex
\Function{TryAdvanceRound}{r}
\If{$\Call{ValidatorsInRound}{r} < n-f $}
    \State \Return \texttt{false}
    \EndIf
\If{$\Call{ValidatorsInRound}{r} > n-f $}
    \State \Return \texttt{true}
\EndIf
\State $consistentSet \gets \Call{TryGetLeaderConsistentSet}{r}$ 
\If{$consistentSet \neq \bot$}
    \State \Return \texttt{true}
\EndIf
\State $leader_{r-1} \gets \Call{GetPredefinedLeader}{r-1}$
\If{$\forall b \in DAG[r],\; b.author \neq leader_{r-1}$}
\State \Return \texttt{true}
\EndIf
\State \Return \texttt{false}
\EndFunction

\Statex 
\Function{TryCreateBlock}{$r, \tau_{\text{enter}}(r)$}
    \If{\Call{LeaderConditionsMet}{$r$} \textbf{ or } 
        $\Call{LocalTime}{\,} \ge \tau_{enter}(r) + \delta_{LT}$ 
        \textbf{ or } 
        \Call{ValidatorsInRound}{$r$} $\ge n-f$}
        
        \State \Call{CreateBlock}{$r$}
        \State \Return true
    \EndIf
    \State \Return false

\EndFunction

\Statex
\Function{LeaderConditionsMet}{$r$}
    \State $leader_{r-1} \gets \Call{GetPredefinedLeader}{r-1}$
    \State $Leaders \gets \{ b \in DAG[r-1] : b.\text{author} = leader_{r-1} \}$
    \State $L1 \gets |Leaders| > 0$

    \State $leader_{r-2} \gets \Call{GetPredefinedLeader}{r-2}$
    \State $Proposals \gets \{ b \in DAG[r-2] \mid b.\text{author} = leader_{r-2} \}$
    \State $Votes \gets DAG[r-1]$

    \State $VoteQuorum \gets \exists b_{leader} \in Proposals \text{ such that }$
    \State \hspace{1.5em} $\left| \{ b.\text{author} \mid b \in Votes \land \Call{IsVote}{b,b_{leader}} \} \right| \ge 2f + p$ %\Comment {quorum of votes}

    \State $SkipPattern \gets \Call{IsSPSkipPattern}{r-2}$
    \State $L2 \gets (VoteQuorum \lor SkipPattern)$

    \State \Return $L1 \land L2$
\EndFunction
\end{algorithmic}
\end{algorithm}

\begin{algorithm}[t]
\caption{DAG construction for validator $v_i$, part 2}
\label{alg:construction2}
\footnotesize
\begin{algorithmic}[1]

\Statex
\Function{TryGetLeaderConsistentSet}{r}
\State $leader_{r-1} \gets \Call{GetPredefinedLeader}{r-1}$
\State $Proposals \gets \{\,b \in DAG[r-1] \mid b.\text{author} = leader_{r-1} \,\}$

\For{$b_{leader} \in Proposals$}
    \State $ConflictingBlocks \gets \Call{GetConflictingBlocks}{b_{leader}}$
    \State $ConflictingVotes \gets \{ b \in DAG[r] \mid \exists b' \in ConflictingBlocks : \Call{IsVote}{b, b'} \}$
    \State $NonConflictingVotes \gets DAG[r] \setminus ConflictingVotes $
    \State $numValidators \gets \left| \{\, b.\text{author} \mid b \in NonConflictingVotes \,\} \right|$
    \If{$numValidators = \Call{ValidatorsInRound}{r}$}
        \State \Return $NonConflictingVotes$
    \EndIf
\EndFor
\If{$|Proposals| = 0$}
    \State \Return $(DAG[r])$
\EndIf
\State \Return $(\bot)$
\EndFunction

\Statex
\Function{SelectParents}{r}
  % \State $leader \leftarrow \textsc{GetPredefinedLeader}(r-2) $
   \State $CandidateParents \gets DAG[r-1]$
   \State $leader_{r-2} \gets \Call{GetPredefinedLeader}{r-2}$
   \If{$\exists b \in DAG[r-1] \text{ s.t. } b.author = leader_{r-2}
    \;\textbf{and}\;
    \Call{ValidatorsInRound}{r-1} = n-f$}
  \State $ConsistentSet \gets \Call{TryGetLeaderConsistentSet}{r-1}$ \Comment{Such a set exists by  \textbf{A1'}.}
%\If{$ConsistentSet \neq \bot$}
    \State $CandidateParents \gets ConsistentSet$
\EndIf
    
    \State $Parents \subseteq CandidateParents$ such that:
\State \hspace{1.5em} (1) for each validator $v$, $Parents$ contains exactly one block authored by $v$, and
\State \hspace{1.5em} (2) if C1 triggered block creation, all blocks in $CandidateParents$ satisfying L1\&L2 are included in $Parents$

    %\If{$ConsistentSet = \bot$}
    \State $Proposals_{r-2} \gets \{u \in DAG[r-2] \mid u.author = leader_{r-2} \}$
    \State $ProposalsWithVotes \gets \{\, u \in Proposals_{r-2}
    \mid \exists u' \in Parents : \Call{IsVote}{u', u} \,\}$
    \If{$|ProposalsWithVotes| > 1$} \Comment{Parent set is not leader-consistent}   
    \State $ExcludedBlocks \gets \{\, b \in DAG[r-1] \mid b.\text{author} = leader_{r-2} \,\}$
    \State  $Parents \gets Parents \setminus ExcludedBlocks$
    \EndIf
    \State \Return $Parents$
\EndFunction

\Statex
\Procedure{CreateBlock}{r}
    \State $b.payload \gets blocksToPropose.dequeueAll() \;\text{if non-empty, else null}$
    \State $b.round \gets r$;
    $b.author \gets v_i$
    \State $P \gets \Call{SelectParents}{r}$
    %\State $P \gets P \cup \{ u \in DAG[1..r-2] \mid \neg \Call{IsPath}{u, b},\ \text{at most one per validator} \}$

    %\For{$r_{old} = r-2$ down to $1$}
    %\For{$u \in DAG[r_{old}]$ such that $\neg \Call{IsPath}{u, b}$}
    \State $U \gets \{\, u \mid u \in DAG[r'],\ 1 \le r' \le r-2 \,\}$
\State \hspace{1.2em} filtered by: $u$ is the latest block from its validator
\State \hspace{1.2em} and $u$ is not referenced by any previous block of $v_i$
        \State $P \gets P \cup U$
    %\EndFor
%\EndFor
    
    \State $b.parents \gets \{HASH(u) \mid u \in P \}$
    \State \Call{SignBlock}{$b$}
    \State \Call{TryAddToDag}{$b$}; 
\EndProcedure

\Statex

\Function{IsValidBlock}{$b$}
    \State \Call{VerifySignature}{$b$}
    \State $P \gets \{ u \in DAG \mid HASH(u) \in b.parents \}$
    \State $P_{r-1} \gets \{ u \in P \mid u.round = b.round-1 \} $
    \If{$|P_{r-1}| < n-f \ \textbf{or} \ \exists u \neq v \in P \mid u.\text{author} = v.\text{author}$}
    \State \Return false
\EndIf
    %\If{$|R| < n-f$ \textbf{or} $\exists u \neq %u' \in R : u.author = u'.author$}
    %\State $R \gets \{ u \in DAG_i[r-1] \mid HASH(u) \in b.parents \}$
    %\If{$|R| < n-f$ \textbf{or} $\exists u \neq u' \in R : u.author = u'.author$}
    %  \State \Return false
    %\EndIf
    \State $leader_{r-2} \gets  \Call{GetPredefinedLeader}{r-2}$
    \If{$\neg \Call{ExposeEquivocation}{b, leader_{r-2}}$}
        \State \Return true
    \EndIf
    \If{$\neg \exists u \in P_{r-1} \mid u.\text{author} = leader_{r-2}$}
    \State \Return true
\EndIf
    \State \Return false
\EndFunction

\Statex
\Procedure{TryAddToDAG}{$b$}
    \State $buffer \gets buffer \cup \{b\}$
    \While{$\exists b' \in buffer \text{ s.t. } \Call{CausalAvailable}{b'} = \texttt{true}$}
    \State choose any $b' \in buffer$ s.t. $\Call{CausalAvailable}{b'} =\texttt{true}$
        \State $buffer \gets buffer \setminus \{b'\}$
        \If{$ \Call{IsValidBlock}{b'}$}
        \State $DAG \gets DAG \cup \{b'\}$ 
        \State $r_{highest} \gets \max(r_{highest}, b'.round)$
        \State $Known[b'.author] \gets Known[b'.author] \cup \Call{GetReachableBlocks}{b'}$
        %\Comment{Update known for the author}
        \EndIf
    \EndWhile
\EndProcedure

\end{algorithmic}
\end{algorithm}

\begin{algorithm}[t]
    \caption{Helper functions}
    \label{alg:consensus-utils}
    \footnotesize

    \begin{algorithmic}[1]

    \Statex
    \Function{GetPredefinedLeader}{$r$}
        \State \Return $(r \bmod n) + 1$
    \EndFunction
    
    \Statex
    \Function{ValidatorsInRound}{$r$}
        \State \Return $\left| \{ b.\text{author} \mid b \in DAG[r] \} \right|$
    \EndFunction

    \Statex
    \Function{CausalAvailable}{$b$}
        \State \Return $\forall h \in b.parents,\ \exists b' \in \text{DAG} \text{ s.t. } \text{HASH}(b') = h$
    \EndFunction

    \Statex
    \Function{IsPath}{$b_{\text{old}}, b_{\text{new}}$}
        \State \Return $\exists\; b_1, \dots, b_k \in DAG,\; k \in \mathbb{N},\ \text{s.t.}$
        \State $b_1 = b_{old},\ b_k = b_{new}
        $ \text{and}
        \State \hspace{1em} $\forall j \in \{2, \dots, k\}:\ \mathrm{HASH}(b_{j-1}) \in b_j.parents$
    \EndFunction

    \Statex
    \Function{GetReachableBlocks}{$b$}
      \State \Return $\{\, b' \in DAG \mid \Call{IsPath}{b',b} = \texttt{true} \,\}$
    \EndFunction

    \Statex
    \Function{IsVote}{$b_{vote}, b_{leader}$}
        \State \Return $\Call{Hash}{b_{leader}} \in b_{vote}.parents$
    \EndFunction

    \Statex
    \Function{IsSPCert}{$b_{cert},b_{leader}$}      
        \State $Parents \gets \{ b \in DAG \mid \Call{HASH}{b} \in b_{cert}.parents \} $
        \State $res \gets |\{b \in Parents \mid \Call{IsVote}{b, b_{leader}}\}|$
        \State \Return $res \geq 2f+p$
    \EndFunction
  
    \Statex
    \Function{IsSPSkipPattern}{$r$}
        \State $leader \gets \Call
        {GetPredefinedLeader}{r}$
        \State $Votes \gets DAG[r+1]$
        \State $Proposals \gets \{ b \in DAG[r] \mid b.author = leader \}$
        \For{$b_{leader} \in Proposals$}
            \State $res \gets \left| \{ b.author \mid b \in Votes,\Call{IsVote}{b,b_{leader}} = \text{false} \} \right|$
            \If{$res < 2f + p$}
            \State \Return \textbf{false}
        \EndIf
    \EndFor

        \If{$|Proposals| = 0 \ \textbf{and} \ \Call{ValidatorsInRound}{r+1} < 2f + p$}
            \State \Return \textbf{false}
        \EndIf
        \State \Return \textbf{true}
    \EndFunction

    \Statex
    \Function{IsDirectlySkippedSlot}{$r$}
        \State $B_{r+2} \gets DAG[r+2]$
        \State $NonFPEvidenceBlocks \gets \{b \in B_{r+2} \mid \Call{IsNonFPEvidence}{b}\}$
        \State  $res_{nonFP} \gets | \{ b.author \mid b \in NonFPEvidenceBlocks \}|$
        \State $FP_{cond} \gets res_{nonFP} \geq  2f + p$
        \State $SP_{cond} \gets \Call{IsSPSkipPattern}{r}$    
    \State \Return $SP_{cond} \ \textbf{and} \ FP_{cond}$
    \EndFunction

    \Statex
    \Function{DirectlyCommittedLeaderBlock}{$r$}
       \State $leader \gets \Call {GetPredefinedLeader}{r}$
        \State $Proposals \gets \{ b \in DAG[r] \mid b.author = leader \}$
        \State $B_{r+1} \gets DAG[r+1]$
        \State $B_{r+2} \gets DAG[r+2]$
        \For{$b_{leader} \in Proposals$}
            \State $res_{SP} \gets \left| \{ b.\text{author} \mid b \in B_{r+2},\ \Call{IsSPCert}{b,b_{leader}} = \text{true} \} \right|$
            \State $res_{FP} \gets \left| \{ b.\text{author} \mid b \in B_{r+1},\ \Call{IsVote}{b,b_{leader}} = \text{true} \} \right|$
            \If{$res_{SP} \ge 2f + p \ \text{or} \ res_{FP} \ge n - p$}
                \State \Return $b_{leader}$
            \EndIf
        \EndFor
        \State \Return \texttt{Undecided}(r)
        \EndFunction

        \Statex
        \Function{IsSPCertifiedPath}{$b_{anchor}, b_{leader}$}
        \State $r \gets b_{leader}.round $
        \State $B_{r+2} \gets DAG[r+2]$
        \State \Return $\exists b \in B_{r+2} :
\Call{IsSPCert}{b, b_{\mathrm{leader}}} \land
\Call{IsPath}{b, b_{\mathrm{anchor}}}$
        \EndFunction
        
    \end{algorithmic}
\end{algorithm}

\begin{algorithm}[t]
    \caption{Helper functions Fast Path}
    \label{alg:consensus-utils1}
    \footnotesize

    \begin{algorithmic}[1]

        \Statex
        \Function{GetConflictingBlocks}{$b_{leader}$}
        \State $B \gets \{b \in DAG[b_{leader}.round] : b.author = b_{leader}.author\ \:\land b \neq b_{leader}\}$
         \State \Return $B$
        \EndFunction

    \Statex

\Function{ExposeEquivocation}{$b$}
    \State $leader_{r-2} \gets \Call{GetPredefinedLeader}{b.round-2}$
    \State $Proposals_{r-2} \gets \{u \in DAG[r-2] \mid u.author = leader_{r-2} \}$
    \State $Parents \gets \{ u \in DAG \mid \Call{HASH}{u} \in b.parents \} $
    \State $ProposalsWithVotes \gets \{ u \in Proposals_{r-2} \mid \exists u' \in Parents : \Call{IsVote}{u', u}$ \}
    \State \Return $|ProposalsWithVotes| > 1$
\EndFunction

    \Statex      
    \Function{IsFPEvidence}{$b_{evid}, b_{leader}$}  
        \State $Parents \gets \{ b \in DAG \mid \Call{HASH}{b} \in b_{evid}.parents \} $
        \State $res_{vote} \gets |\{b \in Parents 
        : \Call{IsVote}{b, b_{leader}}\}| $
        \If{$\neg \Call{ExposeEquivocation}{b_{evid}}$}     
            \State \Return $res_{vote} \geq f+p-1$
        \Else
            \State $B_{conflicts} \gets \Call{GetConflictingBlocks}{b_{leader}}$
            \State $res_{conflict} \gets \left| \left\{ b \in Parents \mid \exists b' \in B_{conflicts}, \Call{IsVote}{b, b'} \right\} \right|$
            \State \Return $(res_{vote} \geq f + p) \wedge (res_{conflict} < f + p)$
        \EndIf
    \EndFunction

    \Statex
    \Function{IsNonFPEvidence}{$b_{evid}$}
        \State $leader_{r-2} \gets \Call{GetPredefinedLeader}{b_{evid}.round-2}$
        \State $Proposals_{r-2} \gets \{u \in DAG[r-2] \mid u.author = leader_{r-2} \}$
        \For{$b_{leader} \in Proposals$}
            \If{$\Call{IsFPEvidence}{b_{evid},b_{leader}}$}
                \State \Return \texttt{false}
            \EndIf
        \EndFor
        \State \Return \texttt{true}
    \EndFunction

    \Statex      
    \Function{PathsToFPEvidenceQuorum}{$b_{anchor},b_{leader}$}
        \State $r \gets b_{leader}.round$
        \State $B_{r+2} \gets DAG[r+2]$
        \State $res \gets \left| \left\{ b.\mathrm{author} \;\middle|\;
    b \in B_{r+2} \land
    \Call{IsFPEvidence}{b, b_{leader}} \land
    \Call{IsPath}{b, b_{anchor}}
\right\} \right|$
        \State \Return $res \geq 2f + p$
    \EndFunction

  \Statex

    \end{algorithmic}
    
\end{algorithm}

\begin{algorithm}[t]
    \caption{Committing and ordering}
    \label{alg:universal-committer}
    \footnotesize

    \begin{algorithmic}[1]
        \State \textbf{Global variables:}
        \State $Sequenced \gets \emptyset$ \Comment{$(id, author)$ pairs whose message was delivered}

    \Statex
    \Procedure{TryDecide}{$r_{decided}$, $r_{highest}$} \label{alg:line:try-decide}
        \State $sequence \gets [\;]$
        \For{$r \in [r_{highest} \text{ down to } r_{decided}+1]$}
            \State $status \gets \Call{TryDirectDecide}{r}$ \label{alg:line:universal:try-direct-decide}
        \If{$status = \texttt{Undecided(r)}$} \label{alg:line:universal:direct-decide-failed}
            \State $status \gets \Call{TryIndirectDecide}{r, sequence}$ \label{alg:line:universal:try-indirect-decide}
        \EndIf
        \State $sequence \gets status || sequence$
        \EndFor
       
        \For{$status \in Sequence$ in ascending round order}
            \If{$status = \texttt{Undecided(r)}$}
                \State \textbf{break}
            \EndIf
            \If{$status = \texttt{Commit}(b_{leader})$}
                \State \Call{OnCommitLeader}{$b_{leader}$}
            \EndIf
            \State $r_{decided} \gets r_{decided} + 1$
        \EndFor       
    \EndProcedure

\iffalse
        \Statex

       \Procedure{TryIndirectDecide\textcolor{red}{Old}}{$c, w, sequence$} \label{alg:line:try-indirect-decide}
        \State $r_{decision} \gets c.\Call{DecisionRound}{w}$
        \State $anchors \gets [s \in sequence \text{ s.t. } r_{decision} < s.round]$
        \For{$a \in anchors$}
        \If{$a = \perp$} \Return $\perp$ \EndIf \label{alg:line:universal:indirect-decide-failed}
        \If{$a = \texttt{Commit}(b_{anchor})$}
        \State $b_{proposer} \gets c.\Call{GetProposerBlock}{w}$
        \If{$c.\Call{CertifiedLink}{b_{anchor}, b_{proposer}}$} \label{alg:line:universal:certified-link}
        \State \Return $\texttt{Commit}(b_{proposer})$ \label{alg:line:universal:indirect-commit}
        \Else
        \State \Return $\texttt{Skip}(w)$ \label{alg:line:universal:indirect-skip}
        \EndIf
        \EndIf
        \EndFor
        \State \Return $\perp$
        \EndProcedure
\fi
        
        \Statex
        \Procedure{TryDirectDecide}{$r$}
            \If{$\Call{IsDirectlySkippedSlot}{r}$}
             \State \Return $\texttt{Skip}(r)$
            \EndIf
        \label{alg:line:baseline:skipped-proposer}
            \State $b_{leader} \gets \Call{DirectlyCommittedLeaderBlock}{r}$
            \If{$b_{leader} \neq \texttt{Undecided}(r)$}
               \State \Return $\texttt{Commit}(b_{leader})$
            \EndIf
            \State \Return $\texttt{Undecided}(r)$
        \EndProcedure
        
    \Statex        
    \Procedure{TryIndirectDecide}{$r, sequence$} \label{alg:line:try-indirect-decide}
    \State $Anchors \gets [s \in sequence \text{ s.t. } s.round > r+2]$
    \For{$a \in Anchors$ in ascending round order}
        \If{$a = \texttt{Undecided(r)}$} 
            \State \Return $\texttt{Undecided(r)}$ 
        \EndIf \label{alg:line:universal:indirect-decide-failed}
        \If{$a = \texttt{Commit}(b_{anchor})$}

        \State $leader \gets \Call {GetPredefinedLeader}{r}$
        \State $Proposals \gets \{ b \in DAG[r] \mid b.author = leader \}$
            
            \State $SortedProposals \gets \Call{Sort}{Proposals}$
            \For{$b_{leader} \in SortedProposals$}
                \State $SP_{cond} \gets \Call{IsSPCertifiedPath}{b_{anchor}, b_{leader}}$
                \State $FP_{cond} \gets \Call{PathsToFPEvidenceQuorum}{b_{anchor}, b_{leader}}$
                \If{
                $SP_{cond}$
                \textbf{or} 
                $FP_{cond}$} \label{alg:line:universal:certified-link} 
                    \State \Return $\texttt{Commit}(b_{leader})$ \label{alg:line:universal:indirect-commit}
                \EndIf
            \EndFor
            \State \Return $\texttt{Skip}(r)$
            \label{alg:line:universal:indirect-skip}
        \EndIf
    
    \EndFor
    \State \Return $\texttt{Undecided(r)}$

\EndProcedure

\Statex
\Procedure{OnCommitLeader}{$b_{leader}$}
    \State $Reachable \gets \textsc{Sort}(\textsc{GetReachableBlocks}(b_{leader}))$
    \For{$b \in Reachable$}
        \State \Call{AppendPayload}{b}
    \EndFor
\EndProcedure

\Statex
\Procedure{AppendPayload}{$b$}
        \If{$b.\text{payload} \neq \text{null}$}
            \For{$(m, id) \in b.payload$}
                \If{$(id, b.author ) \notin Sequenced$}
                    \State $a\_deliver_i(m, id, b.author)$ 
                    \State $Sequenced \gets Sequenced \cup \{(id, b.\text{author})\}$
                \EndIf
            \EndFor
        \EndIf
            
\EndProcedure

    \end{algorithmic}
\end{algorithm}

The pseudocode for FinWhale is given in Algorithms~\ref{alg:construction1}--\ref{alg:universal-committer}. It is based on the Mysticeti design, with the necessary modifications to enable the fast path. Algorithms~\ref{alg:construction1} and~\ref{alg:construction2} describe the construction of the DAG, including executing rounds, creating blocks, advancing rounds, and broadcasting unknown history. Algorithm~\ref{alg:consensus-utils} presents the helper functions, while Algorithm~\ref{alg:consensus-utils1} introduces new helper functions required specifically for enabling the fast path. Algorithm~\ref{alg:universal-committer} describes the logic for deciding leader slots, ordering blocks in the DAG, and delivering payloads.
In the code, we use $DAG[r]$ to refer to the set of all blocks from round $r$ stored in the local DAG.

\section{An Example of FinWhale Execution} \label{sec:run_example}

Figure~\ref{fig:run-example} illustrates an example of the DAG with four validators at a specific moment. 
The protocol iterates the rounds from highest, round 5, to the lowest undecided round, round 1. Initially they are all \texttt{undecided}. The leader of round 5 remains  \texttt{undecided} as there are no blocks in subsequent rounds voting for it.
The leader of round $4$ is directly committed via the fast path, since it has more than $n-p=3$ voters in round $5$.

The leader of round $3$ is directly skipped due to the following two conditions: (a) it has a quorum ($3$) of non-voters in round $4$; and (b) in round $5$, there is a quorum of Non-FP-evidence blocks: each such block references no voters for the round-$3$ leader block and therefore qualifies as a Non-FP-evidence block.

The leader of round $2$ remains \texttt{undecided}. It has only two voters, and is therefore neither directly skipped (which requires at least three non-voters) nor directly committed (which requires at least three voters). Moreover, its anchor in round $5$ is still \texttt{undecided}.

The leader block of round~$1$ cannot be decided directly. It is not directly committed because there is only one vote for it in round~$2$. It is also not directly skipped because, in round~$3$, there are no Non-FP-evidence blocks. All round-$3$ blocks are FP-evidence for the round-$1$ leader block, as each references the round-$2$ block of $V_1$, which votes for that leader block. In the absence of equivocation, referencing $f+p-1=1$ voter is sufficient to qualify as FP-evidence.

However, the leader block of round~$1$ can be indirectly committed. The round-$4$ leader block is marked \texttt{to-commit} and therefore serves as the anchor. Moreover, it has paths to a quorum ($3$) of round-$3$ blocks that are FP-evidence for the round-$1$ leader block. Thus, the round-$1$ leader block is indirectly committed.

The commit sequence is extended to include all leader blocks committed up to the first \texttt{undecided} slot. Therefore, the leader of round~$1$ is included in the commit sequence, all blocks in its causal history are ordered (if they have not already been ordered by a previous leader), and their payloads are delivered.

\begin{figure}[t]
    \centering    \includegraphics[width=0.35\textwidth]{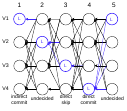}
    \caption{\footnotesize Example of FinWhale execution with four validators. Votes for the leader block are colored in blue.}
    \label{fig:run-example}
\end{figure}
\section{FinWhale's Safety Proof} \label{sec:security}
\label{sec:proof}
In this section, we establish the \emph{safety} of our protocol. Specifically, we prove that the commit sequence is consistent among all honest validators, and consequently that the integrity and total order properties of Byzantine Atomic Broadcast (BAB) hold. To account for the fast path, we adapt existing statements and proofs where necessary, introducing the additional arguments required for fast-path commits while preserving the overall proof structure.

\begin{lemma}
\label{lemma:SP-FP-relation}
Any SP-certificate for block $b$ is also an FP-evidence for $b$.
\end{lemma}
\begin{proof}
Whenever there is an equivocation, an FP-evidence for $b$ includes at least $f+p$ parents voting for $b$ and at most $f+p-1$ voting for any conflicting block $b'$.  
An SP-certificate for $b$ has $\lceil \frac{n+f+1}{2} \rceil = 2f+p$ parents voting for $b$. 
Consequently, there can be at most $f+p-1$ parents voting for any conflicting block $b'$.
In the absence of equivocation, the requirement for an FP-evidence, i.e., at least $f+p-1$ parents voting for $b$, is trivially satisfied.
\end{proof}

\begin{lemma} \label{lemma:recursion-safety_a}
Let $b$ be a leader block in round $r$. If there are 
$\left\lceil \tfrac{n+f+1}{2} \right\rceil = 2f+p$ SP-certificates for $b$ from distinct validators in round $r+2$, 
then every block created in any round $r' > r+2$ (in any DAG) will have a path to an SP-certificate for $b$ from round $r+2$.
\end{lemma}
\begin{proof}
Consider, first, blocks in round $r+3$.
Each such block has at least $n-f$ references to blocks from round $r+2$.  
Since $n-f \geq 2f+p$, by quorum intersection the set of validators whose round-r+2 blocks are referenced intersects the set of validators contributing the 2f+p SP-certificates for b in at least one honest validator.
Hence every block in round $r+3$ references at least one SP-certificate for $b$ from round $r+2$.
Now consider any block in a round $r' > r+3$.  
Its causal history contains at least $n-f$ blocks from round $r+3$, and by the previous case each of these blocks already references an SP-certificate from round $r+2$.  
Therefore, any block in round $r'>r+3$ has a path to an SP-certificate for $b$ from round $r+2$.
\end{proof}

\begin{lemma} \label{lemma:fast-commit-cert}
Let $b$ be a leader block in round $r$.  
If a validator $v_i$ observes $n-p$ blocks from distinct validators voting for $b$ in round $r+1$,  
then every block created in round $r+2$ (in any DAG) will be an FP-evidence for $b$.
\end{lemma}
\begin{proof}
Every block in round $r+2$ references at least $n-f$ blocks from round $r+1$.  
Now suppose validator $v_i$ observed $n-p$ blocks from distinct validators voting for $b$ in round $r+1$:  
\begin{itemize}
    \item If the round-$r+2$ block exposes equivocation by the round-$r$ leader, it does not reference the $r-1$ block of the Byzantine leader, so at most $f-1$ of its parents may be Byzantine. Consequently, it references at least $n-p-f-(f-1) = f+p$ blocks voting for $b$ in round $r+1$, and at most $f+p-1$ blocks voting for a conflicting block. This satisfies the equivocating case of the FP-evidence definition.

    \item If the round-$r+2$ block does not expose equivocation, then among its $n-f$ references it must include at least $n-p-2f=f+p-1$
    blocks voting for $b$ in round $r+1$, which falls under the Non-equivocating case of the FP-evidence definition. 
\end{itemize}
In either case, every block in round $r+2$ qualifies as an FP-evidence for $b$.
\end{proof}

\begin{lemma} \label{lemma:recursion-safety_b}
Let $b$ be a leader block in round $r$.  
If a validator $v_i$ observes $n-p$ blocks from distinct validators voting for $b$ in round $r+1$,
then every block created in round $r'>r+2$ (in any DAG) will link to at least $n-f$ FP-evidence blocks for $b$ from round $r+2$.
\end{lemma}
\begin{proof}   
Every block in a round $r' > r+2$ contains in its causal history at least $n-f$ blocks created by distinct validators in round $r+2$.
By Lemma~\ref{lemma:fast-commit-cert}, each block in round $r+2$ serves as an FP-evidence for block $b$.
Therefore, any block in rounds greater than $r+2$ must reference at least $n-f$ FP-evidence blocks for $b$.
\end{proof}

\begin{lemma} \label{lemma:direct-skip-a}
If an honest validator $v_i$ commits a leader block  $b$ of slot $s$, no honest validator directly skips $s$.
\end{lemma}

\begin{proof}
Suppose, by contradiction, that an honest validator $v_j$ decides to directly skip slot $s$. Then both of the following conditions hold:
\begin{enumerate}
    \item For each leader block of $s$ (if any) in the local DAG of $v_j$, there exists a quorum of $2f+p$ blocks from distinct validators in round $r+1$ that do not vote for that leader block.
    
    \item At least $2f+p$ blocks from distinct validators in round $r+2$ are Non-FP-evidence blocks.
\end{enumerate}

Validator $v_i$ can commit $b$ in one of two ways:

\begin{itemize}
    \item \textbf{Direct commit:}
    
    If $v_i$ observes either $n-p$ blocks voting for $b$ or $2f+p$ SP-certificates for $b$, then in both cases there exists a quorum of blocks from distinct validators voting for $b$.
    
    Any quorum of non-voters obtained by $v_j$ intersects this quorum in at least one honest validator. Hence, $b$ must appear in the local DAG of $v_j$.
    
    Therefore, $v_j$ gathered a quorum of blocks that do not vote for $b$, while $v_i$ observed a quorum voting for $b$, which is impossible by quorum intersection.

    \item \textbf{Indirect commit:}
    
    $v_i$ observes a path from the anchor either to an SP-certificate for $b$, or to a quorum of FP-evidence blocks for $b$ from distinct validators.
    
    \begin{itemize}
        \item If there exists an SP-certificate for $b$, then a quorum of blocks from distinct validators vote for $b$. The contradiction then follows exactly as in the direct commit case.
        
        \item If there exists a quorum of FP-evidence blocks for $b$, then by quorum intersection, at least one of the Non-FP-evidence blocks observed by $v_j$ is itself an FP-evidence block for $b$.
        
        Hence, $b$ appears in the local DAG of $v_j$. Since Non-FP-evidence is defined with respect to all leader blocks of slot $s$ in the local DAG of $v_j$, including $b$, this yields a contradiction.
    \end{itemize}
\end{itemize}
Therefore, if a validator commits a leader block, either directly or indirectly, no other validator can directly skip the slot.
\end{proof}

\begin{lemma} \label{lemma:direct-commit-safety-a}
    Let $b$ be a block of leader slot $s$ in round $r$. If a correct validator $v_i$ directly commits $b$, then no correct validator decides to skip slot $s$.
\end{lemma}
\begin{proof}
By Lemma~\ref{lemma:direct-skip-a}, if $v_i$ directly commits block $b$ of leader slot $s$, no other validator can directly skip $s$.
The remaining case is that a validator attempts to \emph{indirectly skip} $s$. 
This would require that an anchor block in round $r'>r+2$ is marked \emph{to-commit} and satisfies the following two~conditions:
\begin{itemize}
    \item it has no path to to an SP-certificate for $b$, and
    \item it has no paths to any quorum of FP-evidence blocks for $b$ produced by distinct validators.
\end{itemize}

If $v_i$ directly committed $b$ via the slow path, then by Lemma~\ref{lemma:recursion-safety_a} there will be a path from the anchor to an SP-certificate for $b$ in round $r+2$, which contradicts the first~bullet.  
If $v_i$ directly committed $b$ via the fast path, then according to Lemma~\ref{lemma:recursion-safety_b} the anchor block has paths to at least $n-f$ FP-evidence blocks for $b$ (from distinct validators), which contradicts the second bullet above (as $n-f \geq quorum$).
Hence, no correct validator can directly or indirectly skip slot $s$ once $v_i$ has directly committed $b$.
\end{proof}

\begin{lemma} \label{lemma:unique-cert-1}
    For any leader slot $s$, at most one block can gather a quorum ($2f+p$) of votes from distinct validators.
\end{lemma}
\begin{proof}
 Suppose, by way of contradiction, that two conflicting leader blocks gather a quorum of blocks from distinct validators that vote for each of them. The two quorums intersect in at least one honest validator. However, an honest validator cannot vote for two conflicting blocks of the same slot. A contradiction.
\end{proof}

\begin{lemma} \label{lemma:unique-cert_2}
    Consider a block $b$ of slot $s$ that is directly committed via the fast path (i.e., it gathers $n-p$ voters from distinct validators). No other block $b'$ of the same slot $s$ can be committed.
\end{lemma}

\begin{proof}
Suppose, for the sake of contradiction, that there exists another block $b'$ of slot $s$ that is also committed.
We consider all possible ways in which $b'$ can be committed:
\begin{enumerate}
    \item $b'$ is committed directly via either the fast or the slow path, each of which requires at least $2f+p$ votes from distinct validators.

Since $n-p$ blocks from distinct validators vote for $b$, and $n-p \geq 2f+p$, $b'$ cannot gather a sufficient number of votes, by Lemma~\ref{lemma:unique-cert-1}.

\item $b'$ is indirectly committed via a path to an SP-certificate for $b'$.

An SP-certificate for $b'$ requires at least $2f+p$ votes from distinct validators. Since $b$ already gathered at least $2f+p$ votes from distinct validators, this is impossible by Lemma~\ref{lemma:unique-cert-1}.
    \item 
    $b'$ is indirectly committed via a path to a quorum of FP-evidence blocks for $b'$ from distinct validators.

    Since $b$ was directly committed via the fast path by validator $v_i$, any honest validator can produce FP-evidence blocks only for $b$, and not for the conflicting block $b'$, by Lemma~\ref{lemma:fast-commit-cert}.
    \end{enumerate}

In all cases, we reach a contradiction.
    Hence if a leader block is directly committed via the fast path, no other block of the same slot can be committed. 
\end{proof}

\begin{lemma}\label{lemma:unique-cert_3}
Consider a block $b$ of slot $s$ that is directly committed via the slow path. No other block $b'$ of slot $s$ can be committed.
\end{lemma}
\begin{proof}
Suppose, for the sake of contradiction, that there exists another block $b'$ of slot $s$ that is also committed.
We consider all possible ways in which $b'$ can be committed:
\begin{enumerate}
    \item $b'$ is directly committed via either the fast or the slow path.
    
    Since $b$ is directly committed via the slow path, there exist at least $2f+p$ votes for $b$ from distinct validators.
    
    Any direct commit of $b'$ also requires at least $2f+p$ votes from distinct validators, which is impossible by Lemma~\ref{lemma:unique-cert-1}.
    
    \item $b'$ is indirectly committed via a path to an SP-certificate for $b'$.
    
    An SP-certificate for $b'$ requires at least $2f+p$ votes for $b'$ from distinct validators.
    
    Since $b$ is directly committed via the slow path, there already exist at least $2f+p$ votes for $b$ from distinct validators. Hence, by Lemma~\ref{lemma:unique-cert-1}, such an SP-certificate for $b'$ cannot exist.
    
    \item $b'$ is indirectly committed via  paths to a quorum of FP-evidence blocks for $b'$ from distinct validators.

    Since $b$ is directly committed via the slow path, there exist $2f+p$ SP-certificates for $b$ from distinct validators. By Lemma~\ref{lemma:SP-FP-relation}, any set of $2f+p$ SP-certificates for $b$ also constitutes a set of $2f+p$ FP-evidence blocks for $b$.

    Therefore, by quorum intersection, there cannot exist $2f+p$ FP-evidence blocks for a conflicting block $b'$ from distinct validators.
    \end{enumerate}
    In all cases, we reach a contradiction.
    Hence if a leader block is directly committed via the slow path, no other block of the same slot can be committed. 
\end{proof}

Combining Lemma~\ref{lemma:unique-cert_2} and Lemma~\ref{lemma:unique-cert_3} implies the following:
\begin{corollary}
    \label{corollary:unique-commit-a}
    If a validator directly commits a leader block,
    no other correct validator commits a different block for the same slot.
\end{corollary}

\begingroup
\endgroup

\begin{lemma}
\label{lemma:consistent-slots-1}
All honest validators decide a consistent state for each leader slot. Specifically, for any leader slot $s$, if two honest validators decide the state of $s$, then they either both skip $s$ or both commit the same leader block.
\end{lemma}

\begin{proof}
Let $v_i$ and $v_j$ be two honest validators.
Assume, by contradiction, that they made inconsistent decisions for some slot.
Let $s$ be the latest leader slot for which such inconsistent decisions were made, and let $r$ be the corresponding round of $s$.
We consider two cases:

\begin{itemize}
    \item Either $v_i$ or $v_j$ decided the state of $s$ directly.
    
    Suppose one of them decided to directly skip slot $s$. Then, by Lemma~\ref{lemma:direct-skip-a}, the other validator cannot decide to commit any block of $s$.
    
    If one of them decided to directly commit a block $b$ of $s$, then by Lemma~\ref{lemma:direct-commit-safety-a}, the other validator cannot skip $s$.
    
    Moreover, by Corollary~\ref{corollary:unique-commit-a}, no conflicting block of $s$ can be committed.
    
    \item Both $v_i$ and $v_j$ decided indirectly.

    Validator $v_i$ decided according to a committed anchor block in round $r_i$, and validator $v_j$ decided according to a committed anchor block in round~$r_j$.

    By the definition of the anchor, $v_i$ decided to skip all slots in rounds $[r+3, r_i)$ and then committed a leader block in round $r_i$. Similarly, $v_j$ decided to skip all slots in rounds $[r+3, r_j)$ and then committed a leader block in round~$r_j$.

    By the maximality of $r$, the decisions of $v_i$ and $v_j$ are consistent for all slots greater than $r$. Therefore, $r_i = r_j$, and the same leader block was committed in that round.

    Since both validators use the same anchor block, they make the same decisions, as these depend only on the causal history of that block.
\end{itemize}
  In both cases, we arrive at a contradiction to the assumption that $v_i$ and $v_j$ made inconsistent decisions. Therefore, all honest validators decide consistently for every leader slot.  
\end{proof}

\begin{lemma}
\label{lemma:leader-consistent}
All honest validators commit a consistent sequence of leader blocks; that is, the committed leader sequence of any honest validator is a prefix of that of any other honest validator.
\end{lemma}

\begin{proof}
By Lemma~\ref{lemma:consistent-slots-1}, all honest validators make consistent decisions for every leader slot. The commit sequence of a validator contains all committed leader blocks up to the first \texttt{undecided} slot. Therefore, the committed leader sequence of any honest validator must be a prefix of that of any other honest validator.
\end{proof}

\begin{theorem}[Total Order]
    \label{thm:total-order}
    FinWhale satisfies the total order property of Byzantine Atomic Broadcast.
\end{theorem}
\begin{proof}
    Honest validators order the blocks in the DAG according to an identical deterministic sort induced by the causal histories of the committed leader sequence. By Lemma~\ref{lemma:leader-consistent}, all validators have a consistent sequence of committed leader blocks; therefore, all blocks are ordered identically across all validators.

Each ordered block $b$ contributes tuples $(m,id,b.author)$ derived from its payload. Each $(m,id,b.author)$ tuple is delivered exactly once, as duplicate occurrences are filtered out upon delivery. Hence, the delivery sequence is a well-defined projection of the block order onto tuples.

Since all validators apply the same block ordering and the same deterministic filtering rule, the induced tuple delivery order is identical across all validators.

\end{proof}

\begin{theorem}[Integrity]
    \label{thm:integrity}
    FinWhale satisfies the integrity property of Byzantine Atomic Broadcast.
\end{theorem}
\begin{proof}
   The filtering mechanism applied by each honest validator to the  the causal histories of blocks in the committed leader sequence ensures that each $(m,id,\textit{author})$ tuple is delivered exactly once.
\end{proof}

\section{FinWhale Liveness and Fast Termination} \label{sec:more-security}

In this section, we establish liveness under partial synchrony. We first show that any honest leader block is eventually committed via the slow path after GST. Since any block created by an honest validator is part of the causal history of some honest leader block, it follows that every honest block is eventually ordered and its payload is delivered.

In addition, we show that if at most $p$ validators behave Byzantine, an honest leader block can be committed via the fast path in just two rounds (fast termination).

Since liveness is established via the slow path, the proof follows~\cite{mysticeti}, which provides the original liveness argument, and~\cite{starfish}, which further refines and completes the argument, with minor adjustments required for the fast path.

\medskip
For the liveness proof, let $r_{\max}$ denote the largest round reached by any honest validator at GST.

\begin{lemma}[Round-Synchronization] \label{lemma:round synch}
    All honest validators enter any round $r > r_{\max}$ within $\Delta$ time of each other. In addition, they create their round-$r$ blocks within $\Delta$ time of each other.
\end{lemma}
\begin{proof} 
    Let $t_1$ be the time at which the first honest validator advances to round $r > r_{\max}$. By condition $\mathbf{B2}$, this validator broadcasts its unknown history to all validators. Within $\Delta$, all validators receive this history and have at least $n-f$ blocks (from distinct validators) in their DAGs for all rounds $r' < r$.

Any validator $v$ in round $r'$ creates a block according to \textbf{C3}. Consequently, there are at least $n-f+1$ blocks from distinct validators in round $r'$, which implies that \textbf{A1'} holds, while\textbf{A2} holds since $v$ creates its round-$r'$ block. Therefore, $v$ is enabled to advance to the next round.

Assuming local computation time is negligible compared to message delay, by time $t_1 + \Delta$, every honest validator has caught up with all rounds $< r$ and enters round $r$. It follows that all honest validators enter round $r$ within $\Delta$ time of each other.

Regarding synchronization of block creation in round $r$, assume the first honest validator $v_2$ creates a block at time $t_2$. Since $v_2$ is the first honest validator to create a block in round $r$, its block could not have been triggered by \textbf{C3}, as there could be at most $f$ round-$r$ blocks in its DAG (created by Byzantine validators).

If the block is created due to the timeout condition \textbf{C2}, then every other honest validator also creates its round-$r$ block within $[t_2, t_2 + \Delta]$, since all honest validators enter round $r$ within $\Delta$ time of each other and therefore their timeout expirations occur within the same interval.

If $v_2$'s block is created by condition \textbf{C1}, then according to condition \textbf{B1}, $v_2$ broadcasts its history upon block creation. All other validators receive this history by time $t_2 + \Delta$ and can therefore also create their round-$r$ blocks according to \textbf{C1} by that time.

Therefore, all validators create their round-$r$ blocks within $\Delta$ of each other.   
\end{proof}

\begin{lemma}[Timely Delivery] \label{lemma:timely delivery}
    For any round $r > r_{\max}$, all round-$(r-1)$ blocks created by honest validators are integrated into the DAG of every honest validator before its round-$r$ timeout expires.
\end{lemma}
\begin{proof} 
Let $t_{\mathit{GST}}$ be the time of GST, and consider a round $r > r_{\max}$. Let $t_{\mathit{last}}$ be the creation time of the last honest round-$(r-1)$ block. By condition \textbf{B1}, whenever a validator creates a block, it broadcasts both the block and its history. Therefore, every honest validator integrates this block into its DAG by time $t_{\mathit{integrate}} \leq \max(t_{\mathit{GST}}, t_{\mathit{last}}) + \Delta$.

Let $t_v$ be the time at which the first honest validator enters round $r$. By definition of $r_{\max}$, we have $t_v \geq t_{\mathit{GST}}$. By Lemma~\ref{lemma:round synch}, all honest validators enter round $r$ within $\Delta$ time of each other. Hence, the last honest validator enters round $r$ no later than $t_v + \Delta$.

Since an honest validator advances to the next round only after creating its block for the current round (condition \textbf{A2}), it follows that $t_{\mathit{last}} \leq t_v + \Delta$, and therefore $t_v \geq t_{\mathit{last}} - \Delta$.

Combining the inequalities $t_v \geq t_{\mathit{GST}}$ and $t_v \geq t_{\mathit{last}} - \Delta$ with a timeout of $2\Delta$ yields $t_v + 2\Delta \geq \max(t_{\mathit{GST}}, t_{\mathit{last}}) + \Delta \geq t_{\mathit{integrate}}$.

 Thus, every honest round-$(r-1)$ block is integrated into the DAG of every honest validator before the expiration of the round-$r$ timeout.
\end{proof}

\begin{lemma}[leader references]
\label{lemma:leader references}
    If the leader of round $r-1$ is honest, where $r > r_{\max}$, then every round-$r$ block created by an honest validator references (i.e., votes for) the round-$(r-1)$ leader block.
\end{lemma}

\begin{proof}
If a round-$r$ block created by an honest validator $v_i$ was triggered by \textbf{C1}, then by definition it receives the leader block and references it.

If the block was created by condition \textbf{C2}, then by Lemma~\ref{lemma:timely delivery} the honest leader block is received before the timeout, and the block references it.

If the block is created by condition \textbf{C3}, then $v_i$ receives $n-f$ round-$r$ blocks, at least $n-2f$ of which are from honest validators. Consider the set $H$ consisting of the fastest $n-2f$ honest validators to create their round-$r$ blocks. Their block creation is not triggered by $\mathbf{C3}$, as there are not enough round-$r$ blocks to trigger it. Therefore, these blocks are created either by \textbf{C1} or \textbf{C2}, and by the arguments above, they reference the leader block of round $r-1$.

Among the $n-2f$ honest blocks received by $v_i$, at least one belongs to $H$, and thus references the leader block of round $r-1$. Therefore, $v_i$ receives and references the round-$r-1$ leader block.

\end{proof}

\begin{lemma}[quorum of votes]\label{lemma:leader votes}
If the leader of round $r-2$ is honest, where $r > r_{\max}+1$, then every round-$r$ block created by an honest validator references a quorum ($2f+p$) of blocks voting for the round-$(r-2)$ leader block.
\end{lemma}

\begin{proof}
First, observe that if an honest validator creating a round-$r$ block has received $2f+p$ round-$(r-1)$ blocks voting for the leader block of round $r-2$, then by the parent-selection mechanism these blocks are referenced as parents. Indeed, since the leader of round $r-2$ is honest, all round-$(r-1)$ blocks are leader-consistent with respect to that leader. Therefore, no such block is excluded by parent selection.

By Lemma~\ref{lemma:leader references}, every round-$(r-1)$ block created by an honest validator votes for the round-$(r-2)$ leader block. The number of honest validators is $n-f = 2f+2p-1 \geq 2f+p$, for $p \geq 1$.

We now analyze the three block-creation conditions.
If a round-$r$ block created by an honest validator $v_i$ is triggered by \textbf{C1}, then by condition \textbf{L2}, $v_i$ has received $2f+p$ round-$(r-1)$ blocks from distinct validators voting for the round-$(r-2)$ leader block. The \textbf{L2} SP-skip condition cannot hold, since all blocks created by honest validators vote for the leader, so no quorum of non-voting blocks can be formed. Hence, the $2f+p$ voting blocks are selected and referenced as parents.

If the block is created by condition \textbf{C2}, then by Lemma~\ref{lemma:timely delivery}, all round-$(r-1)$ blocks created by honest validators are received by $v_i$ before its round-$r$ timeout expires. Hence, $v_i$ receives at least $2f+p$ round-$(r-1)$ blocks voting for the leader block before timeout expiration. Therefore, the created block references these blocks as parents.

Finally, suppose the block is created by condition \textbf{C3}. Then $v_i$ received $n-f$ round-$r$ blocks from distinct validators, at least $n-2f$ of which are from honest validators.

Consider the set $H$ consisting of the fastest $n-2f$ honest validators to create their round-$r$ blocks. Their block creation cannot be triggered by \textbf{C3}, since there are not yet enough round-$r$ blocks to satisfy that condition. Hence, these blocks are created either by \textbf{C1} or \textbf{C2}, and by the arguments above, each references $2f+p$ blocks voting for the round-$(r-2)$ leader block.

Among the $n-2f$ honest round-$r$ blocks received by $v_i$, at least one block $b$ belongs to $H$. Since $b$ references $2f+p$ blocks voting for the round-$(r-2)$ leader block, these voting blocks are received together with $b$. Therefore, $v_i$ receives and references these blocks as parents as well.
\end{proof}

\begin{lemma}[honest leaders committed]
    \label{lemma:honest leader committed}
    Every leader block created in round $r > r_{\max}$ by an honest validator is committed (i.e., its slot is marked \texttt{to-commit}).
\end{lemma}

\begin{proof}
By Lemma~\ref{lemma:leader references}, every block created by an honest validator in round $r+1$ votes for the honest leader block of round $r$. These $n-f$ voting blocks are eventually received by all honest validators.

If $p = f$, then these $n-f$ blocks from distinct validators are sufficient for a direct fast commit.

If $p < f$, then by Lemma~\ref{lemma:leader votes}, every round-$(r+2)$ block created by an honest validator references a quorum of $2f+p$ voters for the leader block, and thus constitutes an SP-certificate for it. Since $n-f \geq 2f+p$, every honest validator eventually receives at least $2f+p$ such SP-certificates. Consequently, each honest validator marks the slot as \texttt{to-commit}, and the leader block is \emph{committed}.
\end{proof}

\begin{theorem}[Fast Termination]
    \label{thm:fast-termination}
    FinWhale satisfies fast termination.
\end{theorem}
\begin{proof}
    By Lemma~\ref{lemma:leader references}, after GST every block created by an honest validator in round $r+1$ votes for the honest leader block of round $r$. 
    If at least $n-p$ validators behave honestly (equivalently, at most $p$ validators are Byzantine), then the leader's block eventually gathers $n-p$ honest voters, resulting in a direct fast commit after two rounds.
\end{proof}

\begin{lemma}
    \label{lemma:3-consecutive}
     The round-robin leader schedule ensures that in  any window of  $3f+3$ rounds, there are three consecutive rounds with honest leaders. 
\end{lemma}
\begin{proof}
First note that the system consists of $n \ge 3f+1$ validators. In each cycle of $n$ rounds, every validator is scheduled exactly once as leader, and therefore there are at most $f$ Byzantine leaders in each cycle.

Suppose for contradiction that in the cyclic order (including wrap-around), there are no three consecutive rounds whose leaders are honest. Then every maximal block of consecutive honest leaders has length at most two.

In the cyclic order, the $f$ Byzantine leaders act as separators and therefore induce at most $f$ maximal blocks of consecutive honest leaders.
If none of these blocks has length at least three, then each has size at most two, and therefore the total number of honest leaders is at most $2f$, which contradicts the fact that the number of honest leaders is at least $n-f = 2f+1$.
Therefore, the cyclic order contains three consecutive honest leaders (including wrap-around). 

Now consider any window of $3f+3$ consecutive rounds. Such a window contains a full cycle of $n$ rounds plus the first two rounds of the next cycle. Hence it contains every cyclic triple of the leader schedule, including one consisting entirely of honest leaders. Therefore, every such window contains three consecutive rounds with honest leaders. 
\end{proof}

\begin{lemma}[All leader slots decided]
\label{lemma:decision-liveness}
After GST any undecided slot eventually gets decided.
\end{lemma}

\begin{proof}
Assume, for the sake of contradiction, that some slot remains undecided.
Let $r_0$ be a round containing an undecided slot. By
Lemma~\ref{lemma:3-consecutive}, there eventually exist three consecutive
rounds $r',r'+1,r'+2$, with $r'>r_0$, whose leaders are honest. By
Lemma~\ref{lemma:honest leader committed}, the corresponding leader slots
$s',s'_1,s'_2$ are marked \texttt{to-commit}.

Since $r_0<r'$, there is an undecided slot among the rounds up to
$r'+2$. Let $r$ be the highest round up to $r'+2$ containing an
undecided leader slot, and let $s$ denote its slot.

If $r \in {r'-3,r'-2,r'-1}$, then one of the slots
$s',s'_1,s'_2$ is the first committed slot above round $r$ and serves as
the committed anchor for round $r$. Hence, slot $s$ is decided.

If $r < r'-3$, then, by the choice of $r$, all leader slots in rounds
greater than $r$ and up to $r'+2$ are decided. Thus, there is a committed
block in round $r+3$ or higher that can serve as the committed anchor
for round $r$; if no earlier committed slot serves as the anchor, then
the block corresponding to $s'$ serves as the anchor. Hence, slot $s$
is decided.

In both cases, slot $s$ is decided, contradicting the choice of $r$.
\end{proof}

\begin{theorem}[Agreement]
    \label{thm:agreement}
    FinWhale satisfies the agreement property of Byzantine Atomic Broadcast.
\end{theorem}
\begin{proof}
Suppose an honest validator $v_i$
 ordered a block $b$ and delivered its payload. This is the result of $v_i$ applying a deterministic sort to the causal histories of a committed leader sequence $L_0, L_1, \dots, L_n$. Let $v_j$ be any other honest validator.

By Lemma~\ref{lemma:leader-consistent}, the committed leader sequence is consistent across all honest validators. Moreover, by Lemma~\ref{lemma:decision-liveness}, all leader slots are eventually decided, and therefore $L_0, L_1, \dots, L_n$ become part of $v_j$'s commit sequence.

Since $v_j$ applies the same deterministic algorithm to the same committed leader sequence, it will also order $b$ and deliver its payload.
\end{proof}

\begin{lemma}
\label{lemma:commit sequence}
    Every leader block created in round $r > r_{\max}$ by an honest validator is included in the commit sequence.
\end{lemma}
\begin{proof}
 By Lemma~\ref{lemma:honest leader committed}, 
 any leader block $b$ by an honest validator is committed.
 By Lemma~\ref{lemma:decision-liveness} any undecided slot preceding the slot of $b$ will be decided.
Consequently, $b$ will be part of the commit sequence.
\end{proof}

\begin{theorem}[Validity]
    \label{thm:validity}
    FinWhale satisfies the validity property of Byzantine Atomic Broadcast.
\end{theorem}
\begin{proof}
  If an honest validator $v_k$ calls $\mathsf{a\_bcast}_k(m, id)$, then the tuple $(m,id)$ is integrated into the payload of some block $b$ created by $v_k$.

Since every block created by an honest validator becomes part of the causal history of some honest leader block, and by Lemma~\ref{lemma:commit sequence} such a leader block is included in the commit sequence, every honest validator eventually orders $b$ and delivers $(m,id,v_k)$.
    
\end{proof}
\section{Related Work}
\label{sec:related-work}
\subsection{DAG-based Protocols}
DAG-based BFT protocols can be broadly classified into \emph{certified} and \emph{uncertified} DAG protocols, according to the mechanism used to disseminate blocks.
In certified DAG protocols, validators disseminate blocks using a Byzantine Reliable Broadcast (BRB) primitive. 
BRB prevents equivocation and ensures block availability by guaranteeing that blocks are consistently delivered to all honest validators. 
Typically, each block must first be certified by a quorum of validators before being included in the DAG. 
However, BRB-based dissemination may incur an additional latency of $2$--$4$ message delays.
In contrast, uncertified DAG protocols aim to improve latency by using best-effort broadcast (BEB) instead of BRB, where the uniqueness of a block is not guaranteed. 
Thus, these protocols must explicitly handle~equivocations.

We first consider certified DAG-based consensus protocols.
Aleph is an asynchronous protocol that constructs a round-based DAG and uses an efficient binary agreement protocol for ordering.
DAG-Rider~\cite{dag-rider} is an asynchronous DAG-based BFT protocol that progresses in waves consisting of four rounds, with a single leader in each wave.
Tusk~\cite{narwhal-tusk} refines DAG-Rider by replacing conventional BRB with quorum-based BRB, modifying the commit rules to improve latency in common-case executions, and enabling efficient garbage collection.

Bullshark~\cite{bullshark} builds upon DAG-Rider to improve commit latency during synchronous periods. 
Its partially synchronous variant elects one leader every two rounds.
Shoal~\cite{shoal} interleaves two instances of Bullshark, thereby introducing a leader in every round.
Sailfish~\cite{sailfish} and Shoal++~\cite{shoal++} similarly optimize the commit rules for leader blocks by using the first messages of the BRB to count votes for the leader. In addition, they support multiple leaders per round, further improving commit latency.

%Uncertified DAG protocols aim to improve latency by using best-effort broadcast (BEB) instead of BRB, where the uniqueness of a block is not guaranteed. 
%Thus, these protocols must explicitly handle equivocations.

We next discuss uncertified DAG-based protocols. HashGraph~\cite{HashGraph} builds an unstructured DAG through validators disseminating blocks containing two references to previous blocks. It uses virtual voting by interpreting the Hashgraph structure itself, together with the notion of ``strongly seeing'' (similar to certificates), to handle equivocations. A binary agreement protocol is then used to order transactions.

Cordial Miners~\cite{Cordial-Miners} introduce both asynchronous and partially synchronous DAG-based protocols. 
In these constructions, a single leader is assigned in the first round of each wave, where waves consist of five rounds in the asynchronous setting and three rounds in the partially synchronous setting.

Mysticeti~\cite{mysticeti} is a partially synchronous DAG protocol built upon Cordial Miners. 
It pipelines waves and supports multiple leaders per round in order to improve commit latency.
A followup work~\cite{starfish} introduced four DAG protocols:
The first revisits Mysticeti with an improved push-based pacemaker that resolves its liveness issues.
Starfish further decouples payload dissemination from metadata dissemination, improving communication complexity. 
Finally, Starfish-L and Mysticeti-L combine multi-signatures with a Lazy-Push pacemaker to further reduce communication overhead.
In Finwhale, we build on and adapt the revised Mysticeti protocol introduced in~\cite{starfish}, as well as their approach of using a single leader per round.
BBCA-Chain~\cite{bbca-chain} introduces a design in which leader blocks are broadcast using a BBCA primitive built on top of BRB, while non-leader blocks are disseminated using BEB, where leaders are assigned in every round.

\subsection{Fast Consensus}
Kursawe\cite{optimistic-agreement}
was the first to implement a fast (two-step) Byzantine consensus protocol in which  two-step fast path is paired with a subprotocol in the slow path. 
It is able to run with $3f+1$ validators and the fast path is taken when all validators  are honest and the network is synchronous.
Otherwise, a fallback subprotocol is used to ensure liveness.
An asynchronous binary Byzantine consensus protocol with fast-path termination is presented in
\cite{FMR}. The randomized protocol, assisted by a random oracle, tolerates $n>5f$ validators and decides within one communication round in favorable executions. %The paper also proposes a failure-detector-based protocol that tolerate $n>6f$ validators and achieve the same one-round termination property under favorable conditions.\\
%MMR doesnt have a fast path, does it?

Bosco~\cite{bosco} is a Byzantine consensus protocol that achieves a two-step fast path with $n>5f$ validators when all validators propose the same value. 
The fast path can also be achieved with $n>7f$ when all honest validators are in pre-agreement.

FaB Paxos~\cite{fab} is a fast Byzantine consensus protocol requiring $n \ge 5f+1$ validators. 
The paper also presents a parameterized variant with $n \ge 3f+2p+1$, which achieves a two-step fast path in the common case when the leader is correct, the network is synchronous, and at most $p$ validators are Byzantine.

The lower bound of FaB Paxos was later improved to $n \ge 5f-1$ by~\cite{good-case-broadcast}. 
An independent work~\cite{revisiting-optimial-resilience} proved a lower bound of $n \ge 3f+2p-1$ for the parameterized setting. 
The key idea is that equivocating leaders can be detected, allowing validators to wait for $n-f$ votes while excluding the vote of the Byzantine leader.

Banyan~\cite{banyan} is a rotating-leader protocol based on ICC~\cite{ICC}. 
It achieves two-step termination in the partially synchronous model with $n \ge 3f+2p-1$ validators. 
As long as at most $p$ replicas are unresponsive, termination occurs within two message delays. 
The dual-mode mechanism enables two-step finalization latency in the fast path, and ensures that no penalties are incurred when the fast path is not taken.

Kudzu~\cite{Kudzu} is an atomic broadcast protocol with an integrated fast path. 
It achieves finality in just two rounds of communication for $n = 3f + 2p + 1$ validators, provided that at least $n - p$ replicas behave honestly. 
Building on DispersedSimplex~\cite{DispersedSimplex}, Kudzu introduces minimal modifications to incorporate a fast path. 
Furthermore, the protocol uses erasure codes to ensure that the leader can disseminate large blocks with low and well-balanced communication complexity.

In a concurrent work, BlueBottle~\cite{bluebottle} extends Mysticeti with a fast path under the assumption of $n = 5f + 1$ validators. Unlike FinWhale, which targets the parameterized setting and integrates both fast and slow paths while operating under the lower bound of $n = 3f + 2p - 1$, BlueBottle focuses on the non-parameterized setting and the fast path alone.

\section{Conclusions} \label{sec:conclusion}

In this work, we presented FinWhale, the first DAG-based Byzantine Fault Tolerant protocol to integrate a fast-path mechanism that achieves termination in just two message delays when favorable conditions occur. 
By carefully combining fast-path commit patterns with the existing slow-path mechanisms of Mysticeti, FinWhale enables validators to commit transactions more quickly while maintaining safety and liveness under partial synchrony with $n = 3f + 2p - 1$ validators. 

Our protocol tolerates up to $f$ Byzantine faults and ensures fast termination whenever the number of faults does not exceed $p$, all while preserving the foundational guarantees and high-throughput benefits of Mysticeti. 
Beyond achieving the theoretical lower bound for latency in BFT consensus based protocols, FinWhale demonstrates that fast-path techniques traditionally applied to leader-based protocols can be effectively adapted to DAG structures. 
%
%Overall, FinWhale advances the state of the art in DAG-based consensus by providing a formally proved mechanism for low-latency transaction commitment.
This helps bridge the gap between high-throughput DAG designs and optimal fast-path decision-making.

Starfish, Starfish-L, and Mysticeti-L~\cite{starfish} are protocols closely related to Mysticeti and designed to achieve improved communication complexity by combining multi-signatures, lazy dissemination of history, and efficient information dispersal techniques. Exploring how these techniques can be incorporated into FinWhale is an interesting direction for future work.

\paragraph*{Acknowledgements.}
We thank Alberto Sonnino and Sebastian Müller for insightful discussions that helped us better understand Mysticeti and Starfish.

%%
%% Bibliography
%%

%% Please use bibtex, 

\bibliography{references}

%\appendix

\end{document}